\documentclass[10pt,journal,compsoc]{IEEEtran}   

\usepackage[T1]{fontenc}
 \usepackage{url} 
\usepackage{mdframed}
\usepackage{diagbox}  
\usepackage{latexsym,bm}
\usepackage{xcolor}

\usepackage[pdftex]{graphicx}

\usepackage{bbm}
\usepackage{threeparttable}
\usepackage{makecell}
\usepackage{stfloats}
\usepackage{booktabs}
\usepackage{multirow}
\usepackage[ruled,linesnumbered,noend]{algorithm2e}
\usepackage{amsmath}
\usepackage{amssymb}
 \usepackage{amsfonts}
\usepackage{amsthm}
\usepackage{subfigure}
\usepackage{array}
 \usepackage{bigstrut}
\usepackage{cite}

\usepackage{comment}
\newcommand{\Expect}{\mathbb{E}}

\ifCLASSINFOpdf
 
\else

\fi

\begin{document}
\title{%
{Adaptive Heterogeneous  Client Sampling for Federated Learning over Wireless Networks}} %

\author{Bing Luo,~\IEEEmembership{Senior Member,~IEEE,}
        Wenli Xiao,
        Shiqiang Wang,~\IEEEmembership{Senior Member,~IEEE,}\\
        Jianwei Huang,~\IEEEmembership{Fellow,~IEEE,}
        Leandros Tassiulas,~\IEEEmembership{Fellow,~IEEE}%
\IEEEcompsocitemizethanks{
\IEEEcompsocthanksitem Bing Luo is with the Data Science Research Center, Duke Kunshan University, Kunshan, China (e-mail: bing.luo@dukekunshan.edu.cn)
\IEEEcompsocthanksitem Wenli Xiao is with the  Robotics Institute, Carnegie Mellon University, USA. (e-mail: wxiao2@andrew.cmu.edu)
\IEEEcompsocthanksitem Jianwei Huang is with the School of Science and Engineering, Shenzhen Institute of Artificial Intelligence and Robotics for Society, The Chinese University of Hong Kong, Shenzhen, Shenzhen 518172, China (corresponding author, e-mail: jianweihuang@cuhk.edu.cn).
\IEEEcompsocthanksitem Shiqiang Wang  is with IBM T. J. Watson Research Center, Yorktown Heights, NY, USA. (e-mail: shiqiang.wang@ieee.org)
\IEEEcompsocthanksitem Leandros Tassiulas  is with the Department of Electrical Engineering and Institute for Network Science, Yale University, USA. (e-mail: leandros.tassiulas@yale.edu)
\IEEEcompsocthanksitem The research of Bing Luo was supported by the AIRS-Yale Joint
Postdoctoral Fellowship and the Suzhou Frontier
Science and Technology Program (SYG202310). The work of Jianwei Huang is supported by the National Natural Science Foundation of China (Project 62271434), Shenzhen Science and Technology Program (Project JCYJ20210324120011032), Guangdong Basic and Applied Basic Research Foundation (Project 2021B1515120008), Shenzhen Key Lab of Crowd Intelligence Empowered Low-Carbon Energy Network (No. ZDSYS20220606100601002), the Shenzhen Stability Science Program 2023, and the Shenzhen Institute of Artificial Intelligence and Robotics for Society.  The research of Leandros Tassiulas was supported
by the NSF-AoF: FAIN 2132573 and the ARO W911NF-23-1-0088. 
\IEEEcompsocthanksitem This paper
was presented in part at the IEEE INFOCOM 2022~\cite{luo2021tackling} 
}}%

\IEEEtitleabstractindextext{
\begin{abstract}
Federated learning (FL) algorithms usually  sample a fraction of clients  in each round (partial participation) when the number of participants is large and the server's communication bandwidth is limited. Recent works on the convergence analysis of FL have focused on unbiased client sampling, %
e.g., sampling uniformly at random, %
{which suffers from slow wall-clock time for convergence  %
due to high degrees of system heterogeneity (e.g., diverse computation and communication capacities) %
and statistical heterogeneity (e.g., unbalanced and non-i.i.d. data).} %
This paper aims to design an adaptive client sampling algorithm for FL over wireless networks that tackles both system  and statistical heterogeneity  %
to minimize the wall-clock convergence time. %
We obtain a new tractable convergence bound for FL algorithms with arbitrary client sampling probability. %
Based on the bound, %
we analytically establish the  relationship between the total learning time and sampling probability with an adaptive bandwidth allocation scheme, %
which results in  a non-convex optimization problem. %
We design an efficient algorithm %
for learning the unknown parameters in the convergence bound and develop a low-complexity algorithm to approximately solve the non-convex problem. %
Our solution reveals the impact of  system  and statistical heterogeneity parameters on the optimal client sampling design. Moreover, our solution shows that as the number of sampled clients increases, the total convergence time first decreases and then increases because %
    a larger sampling number reduces the number of rounds for convergence  but results in a longer expected  time per-round due to  limited wireless bandwidth.
 Experimental results from both hardware prototype %
 and  simulation %
 demonstrate that  
 our proposed sampling scheme
 significantly reduces the convergence time %
 compared to several baseline sampling schemes. 
 Notably, for EMNIST dataset, our scheme in hardware prototype  %
 spends $71$\% less time than the baseline uniform sampling for reaching the same target loss.
\end{abstract}
\begin{IEEEkeywords}
Federated learning,  client sampling, system heterogeneity, statistical heterogeneity, wireless networks, convergence analysis, optimization algorithm.
\end{IEEEkeywords}}

\maketitle

\IEEEdisplaynontitleabstractindextext

\IEEEpeerreviewmaketitle

\IEEEraisesectionheading{\section{Introduction}}

With the rapid advancement of 5G, Internet of Things (IoT), and social networking applications, there is an exponential increase of data generated at network edge devices, such as smartphones, IoT devices, and sensors \cite{chiang2016fog}. Although these valuable data provide beneficial information for the prediction, classification,  and other intelligent applications, %
in order to analyze and exploit the massive amount of data, standard machine learning  technologies  need to collect the training data {in a central server}.  %
However, such centralized data collection and training fashion can be quite challenging in mobile edge networks due to the limited wireless communication bandwidth  for data transmission 
and data privacy concerns \cite{mao2017survey,park2019wireless,yang2019federated}.

In tackling this challenge, 
{Federated learning (FL)} has recently emerged as an attractive distributed machine learning (DML) paradigm, which 
enables many clients\footnote{Depending on the type of clients, FL can be categorized into cross-device FL and cross-silo FL (clients are companies or organizations, etc.) \cite{kairouz2019advances}. 
We study cross-device FL and use ``device'' and ``client'' interchangeably.}
to collaboratively train a  model under the coordination of a central server while keeping the training data decentralized and private (e.g., \cite{kairouz2019advances,yang2019federated,mcmahan2017communication}). %
Similar to conventional DML systems, FL let the clients  perform most of the computation and a server iteratively aggregate their computed updates. 
Compared to traditional distributed machine learning  techniques,   
FL has two unique features (e.g., \cite{bonawitz2019towards,li2020federated,li2018federated,yu2018parallel,yu2019linear,wang2018adaptive}), as shown in Fig.~1.  {First, 
the training data are distributed 
in a {non-i.i.d.} and {unbalanced} fashion across the clients  (known as \emph{statistical heterogeneity}), which negatively  affects the convergence behavior. Second, clients are massively distributed and with diverse communication and  computation capabilities (known as \emph{system heterogeneity}),} where stragglers can slow down the physical training time. 
{The system heterogeneity is more challenging to analyze when deploying FL in wireless network, %
because the communication bandwidth in wireless environment is limited and shared by all connected devices with potential mutual interference \cite{chen2020convergence,shi2020device,nishio2019client,yang2019scheduling,luo2020cost,luocostJSAC,yang2020energy}.}   

Considering limited communication bandwidth and large communication overhead, FL algorithms (e.g., the de facto FedAvg algorithm in \cite{mcmahan2017communication}) usually %
perform {multiple local iterations} on \emph{a fraction of randomly sampled clients (known as partial participation)}  and then aggregates their resulting local model updates via the central server periodically \cite{mcmahan2017communication, bonawitz2019towards,li2018federated,li2020federated}. 
Recent works have provided theoretical convergence analysis that demonstrates the effectiveness of FL with partial participation    %
in various non-i.i.d. settings %
\cite{%
haddadpour2019convergence,karimireddy2019scaffold,li2019convergence,yang2021achieving, qu2020federated}. %

However, %
these prior works \cite{%
haddadpour2019convergence,karimireddy2019scaffold,li2019convergence,yang2021achieving, qu2020federated} have focused on sampling schemes that are  
uniformly at random or proportional to the clients' data sizes, %
{which often suffer from slow error convergence with respect to wall-clock (physical) time\footnote{We use wall-clock time to distinguish from the number of training rounds. %
} due to high degrees  of the system  and  statistical  heterogeneity}. %
This is because the total FL time {depends on \emph{both the number of training rounds} for reaching the target training  precision and %
\emph{the physical running time in each round}} \cite{stragglers}. 
Although uniform sampling  guarantees that the aggregated model update in each round is unbiased towards that with full client participation, the aggregated model may have a high variance %
due to data heterogeneity, thus,  \emph{requiring more training rounds to converge to a target loss or accuracy precision}. %
Moreover, considering clients'  heterogeneous  computation and communication capacities, uniform sampling also suffers from the \emph{staggering effect}, as the probability of sampling a straggler (slow computation or communication time) within the sampled subset in each round can be relatively  high, %
thus \emph{yielding a long per-round time}.

One effective way of speeding up the error convergence with respect to the number of training rounds is  %
to choose clients according to some sampling distribution where ``important" clients have  high probabilities \cite{zhao2015stochastic,needell2014stochastic,alain2015variance,stich2017safe,gopal2016adaptive}. For example, recent works in FL community have adopted  {importance sampling} approaches  via exploring 
clients' statistical property %
\cite{chen2020optimal,rizk2020federated,nguyen2020fast,cho2020client,pmlr-v139-fraboni21a,fraboni2021impact}.
However, %
a limitation of these works is that they
did not account for the heterogeneous physical time  in each round,  especially under  straggling circumstances. %
Another line of works aims to %
minimize the learning time via optimizing client selection and scheduling based on their heterogeneous system resources when deploying FL in wireless networks
\cite{tran2019federated,chen2020convergence,shi2020device,wan2021convergence,nishio2019client,chai2020tifl, yang2019scheduling,jin2020resource,wang2020optimizing,van2021joint,wang2019adaptive,luo2020cost,tu2020network,wang2021device}. %
However,  {their optimization schemes did not consider
how client selection schemes influence the convergence behavior (e.g., the number of rounds for convergence) due to data heterogeneity and thus affect the total learning time.}%

In a nutshell, the fundamental limitation of existing works is the \emph{lack of joint consideration on the impact of the inherent system heterogeneity and statistical heterogeneity on client sampling}. %
In other words,  clients with {valuable data} may have poor computation and communication capabilities, whereas those who compute or communicate fast may have low-quality data.
This motivates us to study the following key question.

{\setlength{\parskip}{0.3em} \noindent\textbf{Key Question:} \emph{How to design an optimal client sampling scheme for wireless FL that tackles both system and statistical heterogeneity to achieve  fast convergence with respect to wall-clock time?}} %

{\setlength{\parskip}{0.3em}The challenge of the above question is threefold: {(1)  It is difficult to obtain an expected per-round training time for arbitrary client sampling probability due to the \emph{straggling effect} in wireless networks.} (2)  It is challenging to derive an \emph{analytical FL convergence result for arbitrary client sampling probability}. {(3)} The  total convergence time minimization problem can be non-convex  %
and with unknown coefficients, which we can only estimate during the learning process (known as the \emph{chicken-and-egg problem}).}

In light of the above discussion, we
state the main results and key contributions  of this paper
as follows:

\begin{itemize}
\item \emph{Optimal Heterogeneous Client Sampling for FL over Wireless Networks:}   We study how to design the optimal client sampling strategy in wireless networks to minimize FL convergence time.  To the best of our knowledge, this is the first work that aims to optimize the client sampling probability to address both clients' system and statistical heterogeneity.  

\item {\emph{Adaptive Bandwidth Allocation for Arbitrary Client Sampling:} %
Considering the
wireless bandwidth limitation and interference,  we propose an adaptive bandwidth allocation scheme for arbitrary sampled clients to  address the \emph{straggling} effect in each round.
For arbitrary client sampling probability, we characterize the expected time per-round and identify the impact of clients' heterogeneous computation and communication capacities as well as the limited system bandwidth.}

    \item {\emph{Optimization Algorithm with Convergence Guarantee:}   We provide a generic methodology to conduct FL wall-clock time optimization for arbitrary client sampling probability while ensuring convergence.
  A step in this method is to make a minor modification to the existing FedAvg algorithm and its convergence bound. 
    This enables us to formulate an offline   
  non-convex optimization problem 
  with respect to the client sampling probability 
  without actually training the model, where the problem captures both system and statistical heterogeneity. 
   We propose a low-cost substitute sampling approach to learn the convergence-related unknown parameters and develop an efficient algorithm to approximately solve the non-convex problem with low computational complexity.}%

    \item \emph{
    Optimal Client Sampling Principle:}    %
    {Our solution characterizes the impact of system heterogeneity (e.g., clients' computation and communication time) and statistical heterogeneity (e.g., clients' data quantity and quality) on the optimal client sampling probability.} Moreover,  {our solution reveals that as the number of sampled clients increases, the total convergence time first decreases and then increases because %
    a larger sampling number reduces the number of rounds for convergence  but results in a longer expected  time per-round due to  limited wireless bandwidth.}

    \item \emph{Simulation and Prototype Experimentation:} We evaluate the performance of our proposed algorithms through both  %
  a simulated environment and a hardware prototype. %
Experimental results from both real and synthetic datasets  demonstrate that %
our proposed sampling scheme can significantly  reduce  the  convergence  time %
compared to several baseline sampling schemes. %
For example, for hardware prototype with EMNIST dataset,  our sampling scheme spends $71$\% less time than baseline uniform sampling for reaching the same target loss.
{We also show that our proposed sampling performs well even when the strong convexity assumption is violated, e.g., for \emph{non-convex} loss function. We open sourced our experiment code at \url{https://github.com/WENLIXIAO-CS/WirelessFL}}.
\end{itemize}

\begin{table*}[!t]
	\centering
	\caption{Summary of key notations}
	\begin{tabular}{l|l||l|l}
		\toprule[1pt] 
		$F\left( \mathbf{w} \right)$ & Global loss function &
	    	$\mathbf{w}^*$ & Optimal model parameter that minimizes $F\left( \mathbf{w} \right)$  \\
			${F_i}\left( \mathbf{w} \right)$ & Local loss function of client $i$&
		$\mathbf{w}^{(R)}$ & Final model parameter after $R$ rounds\\
		$N$ & Total number of clients&
			$R$ & Final round number for achieving $\epsilon$ \\
					$r$ & Round number index &
			$\epsilon$ & Desired precision with $\Expect[F(\mathbf{w}^{(R)})]-F^{*} \le \epsilon$ \\
	$\boldsymbol{q}$ & Client sampling probability, i.e., $\{q_1,  \ldots, q_N\}$ &
			$K$ & Number of selected clients with $K\!:=\! \left| \mathcal{K}^{(r)}\!\right|\!$ \\
$E$ & Number of local iteration steps	   &
	$\mathcal{K}^{(r)}(\boldsymbol{q})$ & Randomly selected clients in round $r$ with sampling probability  $\boldsymbol{q}$\\
		$T^{(r)}(\boldsymbol{q})$ &  Round time in round $r$ with sampling probability $\boldsymbol{q}$  &
		$T_{\text{tot}}(\boldsymbol{q})$ &  Total learning time under sampling probability  $\boldsymbol{q}$  \\
		$p_i$ &  Client $i$'s datasize weight  &
			$t_i$ &  Client $i$'s  communication time  with unit bandwidth\\
	$f_i^{(r)}$ &  Client $i$'s allocated bandwidth in round $r$  &
				$\tau_i$ &  Client $i$'s computation  time for local model computing \\
				$f_{\text{tot}}$  & Total system bandwidth &
				$G_i$  & {Expected norm of client $i$'s stochastic gradient bound}\\
		$\alpha$, $\beta$ & Unknown constants in %
		convergence bound &
	$R_{\mathbf{q_1},s}$ & Number of rounds for reaching $F_s$ with sampling probability  $\mathbf{q_1}$  \\
		$F_s$ & Pre-defined global loss for estimation $\alpha$ and  $\beta$	&
	$R_{\mathbf{q_2},s}$ & Number of rounds for reaching $F_s$ with sampling probability  $\mathbf{q_2}$  \\
\bottomrule[1pt]
	\end{tabular}%
  \label{keynotation}
\end{table*}%

\section{Related Work}
Active client sampling and selection play a crucial role in addressing the statistical and system heterogeneity challenges in cross-device FL. 
In the existing literature, the research efforts in speeding up the training process mainly focus on two aspects:  %
importance sampling and resource-aware optimization-based approaches.%

The goal of the importance sampling is %
to reduce the high variance in traditional stochastic optimization algorithms of  stochastic  gradient  descent (SGD), where data samples are drawn uniformly at random during the  learning process (e.g., \cite{zhao2015stochastic,needell2014stochastic,alain2015variance,stich2017safe,gopal2016adaptive}). The intuition of importance sampling is to estimate a random variable by seeing important examples more often, but use them less. 
Recent works have adopted this idea in FL systems to improve communication efficiency via reducing the variance of the aggregated model update. %
Specifically, clients with ``important" data would have higher probabilities to be sampled  in each round. For example, existing works use clients' local gradient information (e.g., \cite{chen2020optimal,rizk2020federated,nguyen2020fast}) or local losses (e.g., \cite{cho2020client, pmlr-v139-fraboni21a,fraboni2021impact}) to measure the importance of clients' data.
However, these schemes did not consider the  speed of error convergence with respect to \emph{wall-clock time}, especially the straggling effect due to heterogeneous transmission delays. %

Another line of works aims to minimize wall-clock time via resource-aware optimization-based approaches in mobile edge networks, %
such as CPU frequency allocation (e.g., \cite{tran2019federated}),  and communication bandwidth allocation  (e.g., \cite{chen2020convergence,shi2020device,wan2021convergence}), straggler-aware client scheduling    (e.g., \cite{nishio2019client,chai2020tifl, yang2019scheduling,jin2020resource,wang2020optimizing,van2021joint}), parameters control (e.g., \cite{wang2019adaptive,luo2020cost}), and {task offloading (e.g., \cite{tu2020network,wang2021device})}. 
While these papers provided some novel insights,  %
their optimization approaches %
did not show how client sampling could affect the total convergence time and thus are orthogonal to our work.%

Unlike all the above-mentioned works, our work focuses on how to design the %
optimal client sampling strategy that 
tackles both system and statistical heterogeneity to minimize the  wall-clock  time  with  convergence guarantees. %
In addition, most existing works on FL are based on computer  simulations. In contrast, we implement our algorithm in an actual hardware prototype with resource-constrained devices, which allows us to capture real heterogeneous system operation time. %

The organization of the rest of the paper is as follows.  Section~\ref{sec:systemModel} introduces the system model and problem formulation. Section~\ref{sec:convergence} presents our new error-convergence  bound  with arbitrary client sampling.   Section~\ref{sec:optimizationProblem} gives the  optimal client sampling algorithm and solution insights.  %
 Section~\ref{sec:experimentation} provides the simulation and prototype experimental results. We conclude this paper in Section~\ref{sec:conclusion}.

\section{%
Preliminaries and System Model}
\label{sec:systemModel}

We start by  summarizing the basics of FL and its de facto algorithm FedAvg with unbiased client sampling. Then, we introduce the proposed adaptive client sampling in wireless networks with statistical and system heterogeneity based on FedAvg.  Finally, we present our formulated optimization problem. We summarize all key notations in this
paper in Table~\ref{keynotation}.

\subsection{Federated Learning (FL)}%

Consider a federated learning system involving a set of $\mathcal{N}={1,\ldots, N}$ clients, coordinated by a central server. Each client $i$ has $n_i$ local training data samples ($\mathbf{x}_{i, 1}, \ldots, \mathbf{x}_{i, n_{i}}$), and the total number of training data across $N$ devices is $n_\textnormal{tot} :=\sum\nolimits_{i = 1}^N n_i$. 
Further, define $f(\cdot,\cdot)$ as a loss function where ${{f}\left( \mathbf{w}; \mathbf{x}_{i, j} \right)}$ indicates how the machine learning model parameter $\mathbf{w}$ performs on the input data sample $\mathbf{x}_{i, j}$. Thus, the local loss function of client $i$ can be defined as 
\begin{equation}
\label{lo_ob}
{F_i}\left( \mathbf{w} \right) := \frac{1}{{{n_i}}}\sum\nolimits_{j =1}^{n_i} {{f}\left( \mathbf{w}; \mathbf{x}_{i, j} \right)}.
\end{equation}
Denote $p_i=\frac{n_i}{n_\textnormal{tot}}$ as the weight of the $i$-th device such that $\sum\nolimits_{i = 1}^N p_i=1$. Then, by denoting $F\left( \mathbf{w} \right)$ as the global loss function, the goal of FL is to solve the following optimization problem 
\cite{kairouz2019advances}:
\begin{equation}
\label{gl_ob}
\min_{\mathbf{w}}  F\left( \mathbf{w} \right) :=\sum\nolimits_{i = 1}^N{p_i}{F_i}\left( \mathbf{w} \right).
\end{equation}

The most popular and de facto optimization algorithm to solve \eqref{gl_ob} is FedAvg \cite{mcmahan2017communication}. Here, denoting $r$ as the index of the number of communication rounds, %
we  describe one round (e.g., the $r$-th) of the FedAvg algorithm as follows: 
\begin{enumerate}
    \item The server \emph{uniformly at random samples} a subset of  $K$  clients (i.e., $K\! :=\! \left| \mathcal{K}^{(r)}\right|$ with $\mathcal{K}^{(r)}\! \subseteq\!  \mathcal{N}$) and {broadcasts} the latest model $\mathbf{w}^{(r)}$ to the selected clients. 
    \item{Each sampled client $i$ chooses $\mathbf{w}_i^{(r,0)}\!=\!\mathbf{w}^{(r)}$, and runs $E$ steps\footnote{$E$ is originally defined as epochs of SGD in \cite{mcmahan2017communication}. In this paper, we denote $E$ as the number of local iterations for theoretical analysis.} {of local SGD} on~\eqref{lo_ob} to compute an updated model $\mathbf{w}_i^{(r,E)}$.} 
    
    \item {Each sampled client $i$ lets $\mathbf{w}_i^{(r+1)}\!=\!\mathbf{w}_i^{(r,E)}$ and sends it back to the server.}
     
    \item  The server \emph{aggregates} (with weight $p_i$) the clients' updated model  and computes a new global model $\mathbf{w}^{(\tau+1)}$.%
\end{enumerate}

The above process repeats for many rounds until the global loss converges. %

Recent works have demonstrated the effectiveness of FedAvg with theoretical convergence guarantees  in various settings \cite{%
li2019convergence,haddadpour2019convergence,karimireddy2019scaffold,yang2021achieving, qu2020federated}. {However, these works 
assume that the server samples  clients either uniformly at random or proportional to data size, %
which  may slow down the  wall-clock time for convergence %
due to the straggling effect {e.g., clients’  non-i.i.d. data and diverse computational and communication capabilities}. %
 Thus, a careful client sampling design should tackle both system and statistical heterogeneity for fast convergence.}

\begin{figure}[!t]
	\centering
	\includegraphics[width=9cm,height=6cm]{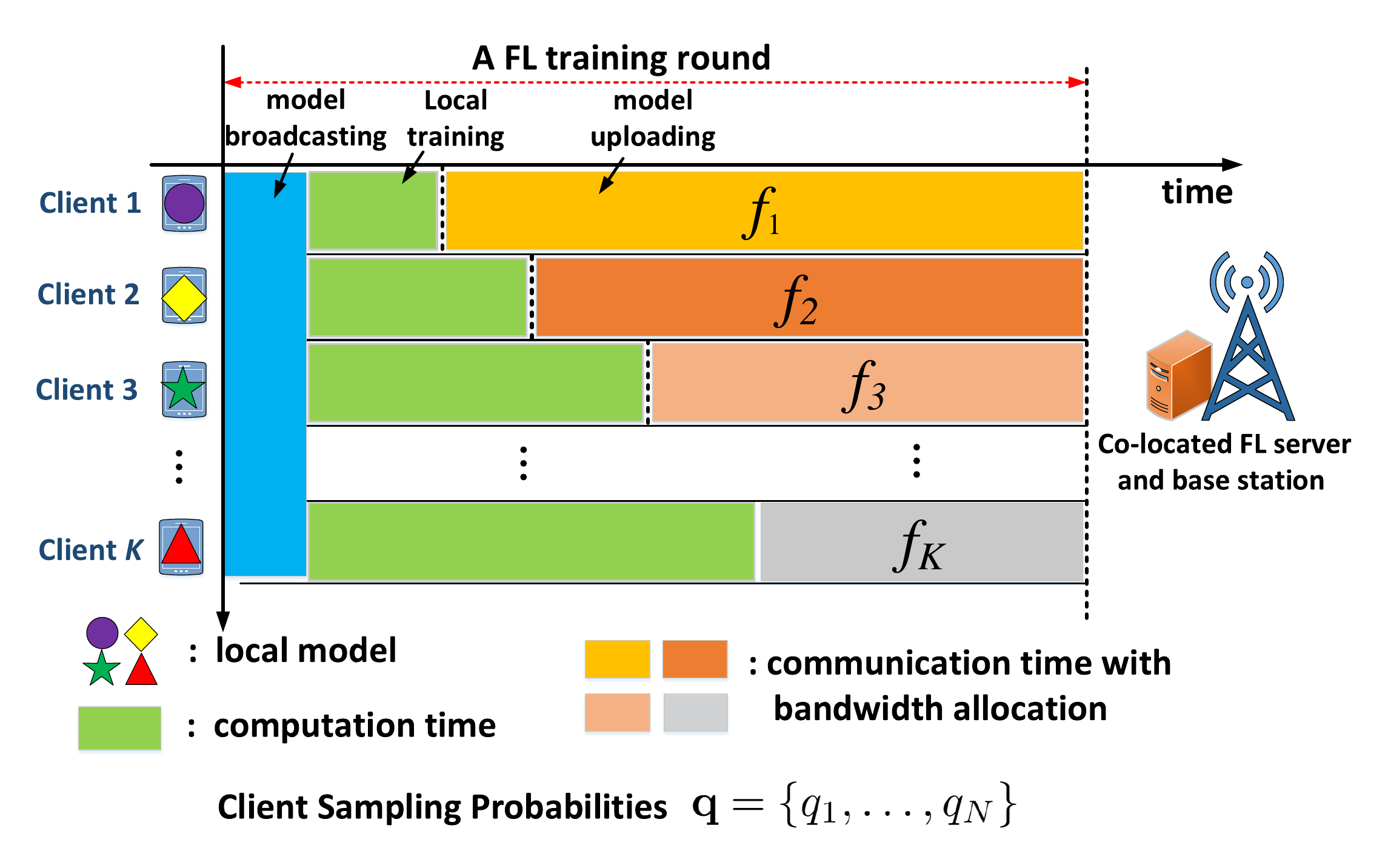}
 	\vspace{-1mm}
	\caption{{A heterogeneous federated learning training round over wireless networks, where $K$ out of $N$ clients are sampled according to the probability distribution  $\boldsymbol{q}=\{q_1, \ldots, q_i, \ldots, q_N\}$, with each sampled client $i$ being allocated bandwidth $f_i$.}}%
	\label{fig:intro2}
\end{figure}

\subsection{{FL Client Sampling over Wireless Networks}}%

{
We aim to sample clients according to a probability distribution 
$\boldsymbol{q}\!=\!\{q_i, \forall i\!\in\!\mathcal{N}\}$, where $0\!\le\! q_i\!\le\!1$   and $\sum_{i=1}^Nq_i\!=\!1$. Through  optimizing $\boldsymbol{q}$, we want to 
address system and statistical heterogeneity %
so as to minimize the wall-clock time for convergence. We describe the system model as follows.
\subsubsection{Client Sampling Model} 
We consider a standard FL setting where the training data are distributed in an unbalanced and  non-i.i.d. fashion among clients. 
Following recent works \cite{li2018federated,  
li2019convergence,haddadpour2019convergence,karimireddy2019scaffold,yang2021achieving, qu2020federated}, {We assume that the server creates the sampled client set $\mathcal{K}^{(r)}(\boldsymbol{q})$ by sampling $K$ times, with replacement, from the total $N$ clients for analytically tractable results. In this scenario, a client may appear multiple times in $\mathcal{K}^{(r)}(\boldsymbol{q})$. The aggregation weight for each client $i$ corresponds to the frequency of its appearance in $\mathcal{K}^{(r)}(\boldsymbol{q})$.}}

{\subsubsection{System Heterogeneity Model}
As illustrated in Fig.~1,  due to system bandwidth limitation and wireless interference, we assume that the sampled clients are scheduled in a frequency-sharing manner,   
and each particular client's  com{p}utation time  for local model updates is static throughout the learning period as \cite{shi2020device, chen2020convergence,yang2020energy}. %
 \emph{Nevertheless, %
each client's  communication time  for model uploading varies across rounds, as the communication time depends on the allocated bandwidth in each round}.\footnote{We do not consider the server's downlink time for model broadcasting and model aggregation, as we mainly focus on the performance bottleneck of the  battery-constrained edge devices.} %
For example, the server may allocate different bandwidth for the same client in different rounds due to the randomly changing  sampled client set. 
}%

{\subsubsection{Joint Computation and Communication Design} 
We consider the mainstream synchronized\footnote{As suggested \cite{bonawitz2016practical,avent2017blender,konevcny2016federated}, we consider mainstream synchronized FL in this paper %
due to its composability with other
techniques (such as secure aggregation protocols and differential privacy).} FL framework (e.g., \cite{mcmahan2017communication, bonawitz2019towards,   
li2019convergence,haddadpour2019convergence,karimireddy2019scaffold,yang2021achieving, qu2020federated}), where the server waits to collect and aggregate all  sampled clients' model updates  before  entering the next round. For  FL operating in the synchronized fashion,  the per-round time is limited by the slowest client (known as straggler).  
\emph{We can prove that for any sampling probability  $\boldsymbol{q}$, we will obtain a minimum round time $T^{(r)}(\boldsymbol{q})$ when the sampled $K$ clients complete their local model computation and model transmission at the same time.\footnote{Otherwise, simply allocating a certain amount of bandwidth from the earlier completed client to the straggler would yield a shorter round time.}} 
In this way, consider sampled client $i$ at round $r$, i.e., $i \in \mathcal{K}^{(r)}(\boldsymbol{q})$, and denote $\tau_i$ as its computation time, $t_i$ as its  communication time when allocated with one unit bandwidth, and $f_i^{(r)}$ as its allocated bandwidth at round $r$. Then, we have  
\begin{equation}
\label{per_round_wireless}
\tau_i+\frac{t_i}{f_i^{(r)}}=T^{(r)}(\boldsymbol{q}), \forall i \in \mathcal{K}^{(r)}(\boldsymbol{q}). 
\end{equation}
Note that for a particular client, the bandwidth allocation can be different across rounds due to the randomly sampled client set in each round. 
For a given round $r$, %
we substitute $f_i^{(r)}$ from \eqref{per_round_wireless} into the total bandwidth constraint  $\sum_{i=1}^Kf_i^{(r)}=f_{tot}$, which leads to
\begin{equation}
\label{total_bandwidth_wireless}
\sum_{i=1}^{K}\frac{t_i}{T^{(r)}(\boldsymbol{q})-\tau_i}=f_{tot}, \forall i \in \mathcal{K}^{(r)}(\boldsymbol{q}).
\end{equation}
However, deriving the analytical solution of $T^{(r)}(\boldsymbol{q})$ from \eqref{total_bandwidth_wireless} is difficult as we can only numerically compute $T^{(r)}(\boldsymbol{q})$ in the non-linear (with the order of $K$) equation. %
}

{Based on the above analysis, we write the total learning time $T_\textnormal{tot}(\boldsymbol{q},R)$ after $R$ rounds %
as %
\begin{equation}
\label{Ttot}
T_\textnormal{tot}(\boldsymbol{q},R)=\sum\nolimits_{r=1}^{(r)}T^{(r)}(\boldsymbol{q}).%
\end{equation}}
Then, we will show the formulated optimization problem in the next subsection.

\subsection{Problem Formulation} 
Our goal is to %
minimize the {expected} total learning time $\Expect[T_\textnormal{tot}(\boldsymbol{q},R)]$, while ensuring that the expected global loss  $\Expect[F\left(\mathbf{w}^{(R)}(\boldsymbol{q})\right)]$ converges to the minimum value $F^{*}$ with an $\epsilon$ precision, 
with $\mathbf{w}^{(R)}(\boldsymbol{q})$ being the aggregated global model after $R$ rounds with client sampling probability $\boldsymbol{q}$. %
This translates into the following problem:
\begin{equation}\begin{array}{cl}
\label{ob1}
\!\!\!\!\!\!\!\textbf{P1:}\quad \min_{\boldsymbol{q}, R} & \Expect[T_\textnormal{tot}(\boldsymbol{q},R)] \\
\quad\quad \text { s.t. } & \Expect[F\left(\mathbf{w}^{(R)}(\boldsymbol{q})\right)]-F^{*} \le \epsilon,\\
&\sum_{i=1}^Nq_i=1, \\
& q_i>0,  \forall i \in \mathcal{N}, %
\ R \in \mathbb{Z}^{+}. 
\end{array}\end{equation}
 The expectation in $\Expect[T_\textnormal{tot}(\boldsymbol{q},R)]$ and $\Expect[F\left(\mathbf{w}^{(R)}(\boldsymbol{q})\right)]$ in \eqref{ob1} is due to the randomness in client sampling probability  $\boldsymbol{q}$ and local SGD. %
Solving Problem \textbf{P1}, however, is challenging in two aspects:
\begin{enumerate}
    \item It is generally impossible to find out how  $\boldsymbol{q}$ and  $R$  affect the final model $\mathbf{w}^{(R)}(\boldsymbol{q})$ and the corresponding loss function  $\Expect[F\left(\mathbf{w}^{(R)}(\boldsymbol{q})\right)]$ before actually training the model. %
    Hence, we need to obtain an analytical expression with respect to $\boldsymbol{q}$ and  $R$  to predict how they affect $\mathbf{w}^{(R)}(\boldsymbol{q})$ and $\Expect[F\left(\mathbf{w}^{(R)}(\boldsymbol{q})\right)]$. %

\item %

{It is difficult to optimize the objective  $\Expect[T_\textnormal{tot}(\boldsymbol{q},R)]$ directly
because the analytical expression of the round time $T^{(r)}(\boldsymbol{q})$  in \eqref{total_bandwidth_wireless} is not available.  Hence, it is difficult to analyze how clients' heterogeneous computation time (i.e., $\tau_i$),  communication time (i.e., $t_i$), and the total bandwidth affect the optimal client sampling.   %
Even for the simplest case, e.g., $K=1$, %
Problem \textbf{P1} can be %
 non-convex  %
 as we will show later.}%

\end{enumerate} 

In Section~\ref{sec:convergence} and Section~\ref{sec:optimizationProblem}, we address these two challenges, respectively, and propose approximate  algorithms to find an approximate solution to %
Problem \textbf{P1} efficiently.   %

\section{Convergence Bound for Arbitrary Client Sampling}
\label{sec:convergence}
In this section, we address the first challenge by %
deriving a new tractable error-convergence bound  for  arbitrary 
client sampling probability. %

\subsection{{Machine Learning Model Assumptions}} To ensure a tractable convergence analysis, %
we first state several assumptions on the local objective functions ${F_i}\left( \mathbf{w} \right)$.
\newtheorem{assumption}{{Assumption}}
\begin{assumption}
L-smooth: For each client $i \in \mathcal{N}$, $F_{i}$ is  $L$-smooth for all $\mathbf{v}$ and $\mathbf{w}$, 
$\|\nabla {F_i}(\mathbf{v})-\nabla {F_i}(\mathbf{w})\| \leq L\|\mathbf{v}-\mathbf{w}\|$.
\end{assumption}
\begin{assumption}
Strongly-convex:  For each client $i \in \mathcal{N}$, $F_{i}$ is $\mu$-strongly convex for all $\mathbf{v}$ and $\mathbf{w}, F_{i}(\mathbf{v}) \geq F_{i}(\mathbf{w})+(\mathbf{v}-$
$\mathbf{w})^{T} \nabla F_{i}(\mathbf{w})+\frac{\mu}{2}\|\mathbf{v}-\mathbf{w}\|_{2}^{2}$.
\end{assumption}
\begin{assumption}
Bounded local variance: %
For each client $i\in \mathcal{N}$ after the $\tau$th iteration, the variance of its  stochastic gradient is bounded: $\mathbb{E}\left\|\nabla F_{i}\left(\mathbf{w}_i^{(\tau)}, \xi_{i}^{(\tau)}\right)\!-\!\nabla F_{i}\left(\mathbf{w}_i^{(\tau)}\right)\right\|^{2} \leq \sigma_{i}^{2}$.
\end{assumption}
Assumptions 1--3 are %
common in many existing studies of  convex FL problems, such as  $\ell_{2}$-norm regularized linear regression, logistic regression (e.g., \cite{li2019convergence,yu2018parallel,chen2020optimal,stich2018local,qu2020federated,cho2020client}).  %
Nevertheless, the experimental results to be presented
in Section VI show that our results also works well for \emph{non-convex}
function such as neural network.

\subsection{{Aggregation with Arbitrary Client Sampling Probability}} 
This section shows how to aggregate clients' model update under sampling probability $\boldsymbol{q}$, such that the aggregated global model is unbiased comparing with that with full client participation, which leads to our  convergence result. 

We first define the \emph{virtual weighted aggregated model with full client participation} in round $r$ as
\begin{equation}
    \label{full_sample}
    \overline{\mathbf{w}}^{(r+1)}:=\sum\nolimits_{i=1}^Np_iw_i^{(r+1)}.
\end{equation}
With this, we can derive the following result. \newtheorem{lemma}{{Lemma}}
\begin{lemma}
\label{adaptive_sam_agg}
\textbf{(Adaptive Client Sampling and Model Aggregation)} %
When clients $\mathcal{K}^{(r)}(\boldsymbol{q})$ are sampled with probability $\boldsymbol{q}=\{q_1, \ldots q_N\}$ and their local updates 
are aggregated as {\begin{equation}
    \label{aggregation}
 \mathbf{w}^{(r+1)} \leftarrow \mathbf{w}^{(r)}+\sum_{{j} \in \mathcal{K}^{(r)}(\boldsymbol{q})} \frac{p_{{j}}}{K q_{{j}}} \left(\mathbf{w}_{{j}}^{(r+1)}-\mathbf{w}^{(r)}\right),
\end{equation}}
then we have  
\begin{equation}
    \label{unbiased_agg}
    \Expect_{\mathcal{K}^{(r)}(\boldsymbol{q})}[\mathbf{w}^{(r+1)}]= \overline{\mathbf{w}}^{(r+1)}.
\end{equation}
\end{lemma}
\begin{proof}
Substituting \eqref{aggregation} into \eqref{unbiased_agg}, we have
{\begin{equation}
\label{prooflemma2}
\begin{array}{cl}
\Expect_{\mathcal{K}^{(r)}(\boldsymbol{q})}\!\!\left[\mathbf{w}^{(r\!+\!1)\!}\right]\!\!\!\!\!\!&=\!\mathbf{w}^{(r)}\!+\!\Expect\left[\sum\limits_{{j} \in \mathcal{K}^{(r)}(\boldsymbol{q})} \frac{p_{{j}}}{K q_{{j}}} \left(\mathbf{w}^{(r+1)}_{{j}}\!-\!\mathbf{w}^{(r)}\right)\right]
\\
&=\!\mathbf{w}^{(r)}\!+\!\frac{1}{K}\Expect\left[\sum\limits_{{{j}} \in \mathcal{K}^{(r)}(\boldsymbol{q})} \frac{p_{{j}}}{q_{{j}}} \left(\mathbf{w}^{(r\!+\!1)}_{{j}}\!-\!\mathbf{w}^{(r)}\right)\right]\\
&=\mathbf{w}^{(r)}\!+\! \frac{K}{K}\Expect_{{j}\in \mathcal{K}^{(r)}(\boldsymbol{q})}\left[\frac{p_{{j}}}{q_{{j}}}\left(\mathbf{w}^{(r\!+\!1)}_{{j}}\!-\!\mathbf{w}^{(r)}\right)\right]\\
&=\mathbf{w}^{(r)}\!+\!\sum\limits_{i=1}^Nq_i\frac{p_i}{q_i}\left(\mathbf{w}^{(r+1)}_{i}\!-\!\mathbf{w}^{(r)}\right)\\
&=\mathbf{w}^{(r)}\!+\!\sum_{i=1}^Np_i\left(\mathbf{w}^{(r+1)}_{i}\!-\!\mathbf{w}^{(r)}\right)\\
&=\mathbf{w}^{(r)}\!+\!\overline{\mathbf{w}}^{(r+1)}-\mathbf{w}^{(r)}=\overline{\mathbf{w}}^{(r+1)}.
\end{array}
\end{equation}}
\end{proof}

\textbf{Remark}: {The interpretation
of our {sampling} and {aggregation} is similar to that of \emph{importance sampling}. More specifically, {since we sample different clients with different probabilities (e.g., $q_i$ for client $i$),  \emph{we need to inversely re-weight  their {updated model's gradient}\footnote{{{Note that simply inversely weighted the model updates from the sampled clients does not yield an unbiased global model, i.e. $\Expect_{\mathcal{K}(\boldsymbol{q})^{(r)}}[\sum_{{j} \in \mathcal{K}(\boldsymbol{q})^{(r)}} \frac{p_{{j}}}{K q_{{j}}}\mathbf{w}_{{j}}^{(r+1)}] \neq \overline{\mathbf{w}}^{(r+1)}$. This is because the equality holds only when clients are sampled uniformly at random.}}}
in the aggregation step (e.g.,  ${1}/{q_i}$ for client $i$), such that the aggregated  model is still unbiased towards that with full client participation}.}} 
We summarize how  the  server  performs  client  sampling  and model  aggregation in 
Algorithm~\ref{FL_sample_agg}, where the main differences compared to the de facto FedAvg  in \cite{mcmahan2017communication} are the \emph{Sampling} (Step~\ref{alg:adaptivefedavgStep1}) and \emph{Aggregation} (Step~\ref{alg:adaptivefedavgStep4}) procedures. Notably, Algorithm~\ref{alg:adaptivefedavg} {recovers variants of the FedAvg algorithm as special cases as in reference [23]. For instance, when $q_i$ is set to $\frac{1}{N}$, it means uniform sampling with replacement. If $q_i$ is set to $p_i$, it refers to weighted sampling with replacement.}

\begin{algorithm}[t]
\label{FL_sample_agg}
\small
	\caption{FL with Arbitrary Client Sampling}
	\label{alg:adaptivefedavg}
	\KwIn{Sampling probability $\boldsymbol{q}=\{q_1, \ldots, q_N\}$, $K$, $E$, precision $\epsilon$, initial model $\mathbf{w_0}$}
	\KwOut{Final model parameter $\mathbf{w}^{(R)}$}
	\For{$r=0,1,2,..., R$%
	}{	\emph{Server randomly samples a subset of clients  $\mathcal{K}^{(r)}(\boldsymbol{q})$ {according to}   $\boldsymbol{q}$, and sends current global model   $\mathbf{w}^{(r)}$ to the selected clients\label{alg:adaptivefedavgStep1}\tcp*{\textbf{Sampling}}}

Each sampled client $i$ %
lets $\mathbf{w}_i^{(r,0)}\!=\!\mathbf{w}^{(r)}$, and  performs %
$\mathbf{w}_i^{(r,j+1)} \!\leftarrow\! \mathbf{w}_i^{(r,j)}-\eta^{(r)} \nabla F_{k}\left(\mathbf{w}_i^{(r,j)}, \xi_i^{(r,j)}\right),j\!=\!0,1,\dots, E\!-\!1$,
and lets $\mathbf{w}_i^{(r+1)}=\mathbf{w}_i^{(r,E)}$		\label{alg:adaptivefedavgStep2}\tcp*{Computation}
		
	Each sampled client $i$ %
	sends back updated model $\mathbf{w}_i^{(r+1)}$ to the server
	\label{alg:adaptivefedavgStep3}\tcp*{Communication}

\emph{Server computes a new global model parameter as  $\mathbf{w}^{(r+1)} \leftarrow \mathbf{w}^{(r)}+\sum_{i \in \mathcal{K}(\boldsymbol{q})^{(r)}} \frac{p_{i}}{K q_{i}} \left(\mathbf{w}_i^{(r+1)}-\mathbf{w}^{(r)}\right)$ \label{alg:adaptivefedavgStep4}\tcp*{\textbf{Aggregation}}}

	}
\end{algorithm}

\subsection{{Convergence Result for Arbitrary Client Sampling}} 
Based on %
Lemma~\ref{adaptive_sam_agg}, we present the main convergence result for arbitrary client sampling probability $\boldsymbol{q}=\{q_1, \ldots, q_N\}$ and required number of rounds $R$ in Theorem~\ref{convergencebound}.

\newtheorem{theorem}{{Theorem}}
\begin{theorem}\label{convergencebound}
\textbf{(Convergence Upper Bound)} Let Assumptions 1 to 3 hold, %
$\gamma=\max \{\frac{8 L}{\mu}, E\}$, decaying learning rate $\eta_{r}=\frac{2}{\mu(\gamma+r)}$, {and $G_i=\max_{r\in \{1,\ldots R\}}\left\|\nabla F_{i}\left(\mathbf{w}_{i}^{(r)}, \xi_{i}^{(r)}\right)\right\|$ denoting client i’s maximum gradient norm throughout the training rounds}. %
For given client sampling probability $\boldsymbol{q}=\{q_1, \ldots, q_N\}$ and the corresponding aggregation in Lemma~\ref{adaptive_sam_agg}, the global loss error after $R$ rounds satisfies %
\begin{equation}
    \label{convergence}
 \Expect[F\left(\mathbf{w}^{(R)}(\boldsymbol{q})\right)]-F^{*} \le\frac{1}{R}\left(\alpha{\sum_{i=1}^N\frac{p_i^2G_i^{2}}{Kq_i}} +\beta\right),
\end{equation}
where $\alpha\!=\!\frac{8LE}{\mu^2}$ and $\beta\!=\!\frac{2L}{\mu^2E}B+\frac{12L^2}{\mu^2E}\Gamma+ \frac{4L^2}{\mu E}\left\|\mathbf{w}_{0}\!-\!\mathbf{w}^{*}\right\|^{2}$, with $B\!=\!\sum\limits_{i=1}^{N} p_{i}^{2} \sigma_{i}^{2}\!+\!8\!\sum\limits_{i=1}^Np_iG_i^2E^{2}\!$ and $\Gamma\!=\!F^{*}\!-\!\sum_{i=1}^{N} p_{i} F_{i}^{*}$.  %
\end{theorem}
\begin{proof}
{
The proof includes three main steps. First, we present the convergence result for full client participation. Then, we show the model variance between the proposed sampling scheme and full client participation. Finally, we use show the expected loss error of between the proposed sampling scheme and full client participation using  L-smooth, which leads to the convergence result of proposed sampling scheme. We elaborate the details as following.}  

{First of all, following the similar proof of convergence under full client participation, we have the following result from the Theorem 1 in \cite{stich2018local,li2019convergence}: 
\begin{equation}
\begin{aligned}
\label{full_conv}
\Expect[F\left(\overline{\mathbf{w}}^{(R)}\right)]-F^{*} \le\frac{\beta}{R},
\end{aligned}
\end{equation}
where $\Expect[F\left(\overline{\mathbf{w}}^{(R)}\right)]$ is the  expected global loss after $R$ rounds with full client participation, and $\beta$ is the same as in \eqref{convergence}.} 

{Then, based on the derived result of Lemma $1$ for client sampling probabilities  $\boldsymbol{q}$, the expected aggregated global model $\Expect_{\mathcal{K}(\boldsymbol{q})^{(r)}}[\mathbf{w}^{r+1}]$ is unbiased compared to full participation $\overline{\mathbf{w}}^{r+1}$. Hence, we can show that the expected difference  of the two (sampling variance) is bounded as follows: 
{\begin{equation}
\begin{aligned}
       \label{bounded_variance}
   \operatorname{Var_{samp}}&= \mathbb{E}_{\mathcal{K}(\mathbf{q})^{(r)}}\left\|\mathbf{w}^{(r+1)}-\overline{\mathbf{w}}^{(r+1)}\right\|^{2}\\
    & \leq \frac{1}{K} \sum_{i=1}^N  \frac{p_i^2}{q_i}  \mathbb{E}\left\|\sum_{j=1}^{E} \eta_{r,j}\nabla F_i\left(\mathbf{w}_i^{(r,j)}, \xi_i^{(r,j)}\right)\right\|^2 \\
&\leq \frac{4\eta_{r, E}^2 E^2}{K} \sum_{i=1}^N  \frac{p_i^2G_i^{2}}{q_i}, 
\end{aligned}
\end{equation}
\noindent where $\mathbf{w}^{(r+1)}$ and $\overline{\mathbf{w}}^{(r+1)}$ are described in Lemma~1, and $\eta_{(r,j)}$ is the non-increasing learning rate for round $r$ at step $j$, with $\eta_{r,1}\le2\eta_{r,E}$.}} 

{Next, we use L-smooth to bound on  $\Expect[F\left(\mathbf{w}^{(r)}(\boldsymbol{q})\right)]- \Expect[F(\overline{\mathbf{w}}^{(r)})]$ as follows:
\begin{equation}
 \label{parti_full}
    \begin{aligned}
    \Expect[F\!\left(\!\mathbf{w}^{(r)}(\boldsymbol{q})\!\right)]-& \Expect[F(\overline{\mathbf{w}}^{(r)})] \leq \frac{L}{2}\mathbb{E}_{\mathcal{K}(\mathbf{q})^{(r)}}\left\|\mathbf{w}^{(r+1)}\!-\!\overline{\mathbf{w}}^{(r+1)}\right\|^{2}\\
    &\leq \frac{L}{2}\frac{4}{K}\sum_{i=1}^N\frac{p_i^2G_i^2}{q_i}(\eta^{r-1}E)^2\\
    &=\frac{2LE^2}{K}\frac{4}{\mu^2(\gamma-1+r)^2}\sum_{i=1}^N\frac{p_i^2G_i^2}{q_i}\\
    &\leq\frac{8LE^2}{\mu^2KrE}\sum_{i=1}^N\frac{p_i^2G_i^2}{q_i}\\
    &=\frac{8LE}{\mu^2Kr}\sum_{i=1}^N\frac{p_i^2G_i^2}{q_i}.
    \end{aligned}
\end{equation}}
{Finally, by letting $r=R$ in equation \eqref{parti_full} and adding \eqref{full_conv},  we have 
\begin{equation}
  \Expect[F\left(\mathbf{w}^{(R)}(\boldsymbol{q})\right)]-F^{*} \le\frac{1}{R}\left(\alpha{\sum_{i=1}^N\frac{p_i^2G_i^{2}}{Kq_i}} +\beta\right),
\end{equation}
which concludes the proof of Theorem 1. We observe that the main difference of the contraction bound compared to full client participation is the sampling variance in \eqref{bounded_variance}, which yields the additional term   $\alpha{\sum_{i=1}^N\frac{p_i^2G_i^{2}}{q_i}}$ in \eqref{convergence}.}
\end{proof}
We summarize the key insights of Theorem~1 as follows:

\begin{itemize}
   \item 
The convergence bound in \eqref{convergence} establishes the relationship between client sampling probability $\boldsymbol{q}$ and the number of rounds $R$ for the expected loss $\Expect[F\left(\mathbf{w}^{(R)}(\boldsymbol{q})\right)]$ reaching the target $\epsilon$  precision. %
 Notably, the derived bound %
 generalizes  the convergence results in \cite{li2019convergence}, where clients are uniformly sampled ($q_i={1}/{N}$) or weighted sampled ($q_i=p_i$). %
\item The convergence bound in \eqref{convergence}   indicates that in order to obtain an unbiased global model, all clients need to be sampled 
with non-zero probability for model convergence, i.e., $q_i>0$, for all client $i$. This is because when $q_i\rightarrow0$, it will take infinite number
of rounds for convergence.

\item The convergence bound in \eqref{convergence} characterizes the impact of clients' heterogeneous data, e.g., $p_iG_i$ on the convergence rate.  
  \end{itemize}

 \section{Adaptive Client Sampling Algorithm}
\label{sec:optimizationProblem}
In this section, {we first %
obtain the analytical lower and upper bound of the expected total learning time $\Expect[T_\textnormal{tot}]$ with sampling probability $\boldsymbol{q}$ and training round $R$.} %
Then, %
we formulate an %
 analytical approximate problem of the original Problem \textbf{P1} based on the convergence upper bound in %
Theorem~\ref{convergencebound}.  %
Finally, we develop an efficient algorithm to solve the new problem with insightful sampling principles. %

{
Here, we note that any optimization for minimizing the FL training time has the challenge that, before the model has been fully trained, it is generally impossible to know exactly how different FL configurations affect the FL performance. Therefore, to optimize the FL performance, we need to have a lightweight surrogate that can (approximately) predict what will happen if we choose a specific configuration. We use the convergence upper bound as the surrogate for this purpose, which is a common practice~\cite{chen2020convergence,yang2020energy,tran2019federated,wan2021convergence,wang2019adaptive}.
}

\subsection{Analysis for the expected round time  {$\Expect[T^{(r)}(\boldsymbol{q})]$}}
{Since the analytically characterizing the expected round time  $\Expect[T^{(r)}(\boldsymbol{q})]$
in \eqref{total_bandwidth_wireless} is infeasible,  we consider the following approximation to characterize how clients' heterogeneous computation time (i.e., $\tau_i$) and heterogeneous communication time (i.e., $t_i$)  affect the optimal client sampling.}
 
{Without loss of generality, we order $N$ clients such that  their computation time are in the ascending order as follows: %

\begin{equation}
\label{reorder_tau}\tau_{1} \leq \tau_{2} \leq \ldots \leq \tau_{i}  \leq \ldots \leq \tau_{N}.
\end{equation}
Then, we show the following theorem that characterizes the lower and upper bounds of   $\Expect[T^{(r)}(\boldsymbol{q})]$.
\begin{theorem}
\label{lemma:expect_tr}
The expected round time 
$\Expect[T^{(r)}(\boldsymbol{q})]$ is lower and upper bounded by %
\begin{equation}
\begin{aligned}
\label{E_T_total}
\frac{{K}\sum\nolimits_{i=1}^{N}q_it_i}{f_{tot}}+\Expect\left[\min_{i \in  \mathcal{K}^{(r)}(\boldsymbol{q})}\left\{\tau_{i}\right\}\right] \le \Expect[T^{(r)}(\boldsymbol{q})]  \\
  \le \frac{{K}\sum\nolimits_{i=1}^{N}q_it_i}{f_{tot}}+\Expect\left[\max_{i \in \mathcal{K}^{(r)}(\boldsymbol{q})}\left\{\tau_{i}\right\}\right],
   \end{aligned}
\end{equation}
where $\Expect\left[\min_{i \in \mathcal{K}^{(r)}(\boldsymbol{q})}\left\{\tau_{i}\right\}\right]$ and $\Expect\left[\max_{i \in \mathcal{K}^{(r)}(\boldsymbol{q})}\left\{\tau_{i}\right\}\right]$ are the expected fastest and slowest client's computation time in the sampled client set $\mathcal{K}^{(r)}(\boldsymbol{q})$, whose analytical expressions are as follows:  
\begin{equation}
\label{Eperround_T_lower}
\Expect\left[\min_{i \in \mathcal{K}^{(r)}(\boldsymbol{q})}\left\{\tau_{i}\right\}\right]=\sum\limits_{i=1}^{N}\!\left[\left(\sum\limits_{j=i}^{N}\! q_{j}\right)^{\!K}\!\!-\!\left(\sum\limits_{j=i+1}^{N} \!q_{j}\right)^{\!K}\right]  \tau_{i},
\end{equation}
 \begin{equation}
\label{Eperround_T_upper}
\Expect\left[\max_{i \in \mathcal{K}^{(r)}(\boldsymbol{q})}\left\{\tau_{i}\right\}\right]=\sum_{i=1}^{N}\!\left[\left(\sum\limits_{j=1}^{ i}\! q_{j}\right)^{\!K}\!\!-\!\left(\sum\limits_{j=1}^{i-1} \!q_{j}\right)^{\!K}\right]  \tau_{i}.
\end{equation}
\end{theorem}%
\begin{proof}
For any sampled $\mathcal{K(\boldsymbol{q})}^{(r)}$ clients in round $r$, we have the lower and upper bounds of \eqref{total_bandwidth_wireless} as 
\begin{equation}
\begin{aligned}
\label{lower_upper_bound_wireless}
\frac{\sum_{i=1}^{K}t_i}{T^{(r)}(\boldsymbol{q})-\min_{i \in \mathcal{K}^{(r)}(\boldsymbol{q})}\left\{\tau_{i}\right\}}<\sum_{i=1}^{K}\frac{t_i}{T^{(r)}(\boldsymbol{q})-\tau_i}\\
<\frac{\sum_{i=1}^{K}t_i}{T^{(r)}(\boldsymbol{q})-\max _{i \in \mathcal{K}^{(r)}(\boldsymbol{q})}\left\{\tau_{i}\right\}}.
  \end{aligned}
\end{equation}
Then, substituting \eqref{total_bandwidth_wireless} into the above inequality, we have 
\begin{equation}
\label{Tr_bound_wireless}
\frac{\sum_{i=1}^{K}t_i}{f_{tot}}+\min_{i \in \mathcal{K}(\boldsymbol{q})^{(r)}}\left\{\tau_{i}\right\}<{T^{(r)}}<\frac{\sum_{i=1}^{K}t_i}{f_{tot}}+\max _{i \in \mathcal{K}(\boldsymbol{q})^{(r)}}\left\{\tau_{i}\right\}.
\end{equation}

Since computation time $\tau_i$ and communication time $t_i$ is independent, taking the expectation of $\frac{\sum_{i=1}^{K}t_i}{f_{tot}}$ in \eqref{Tr_bound_wireless} over sampled clients $\mathcal{K}^{(r)}(\boldsymbol{q})$, we have  
\begin{equation}
\label{Eperround_ti}
\Expect_{i \in \mathcal{K}^{(r)}(\boldsymbol{q})}\left[\frac{\sum_{i=1}^{K}t_i}{f_{tot}}\right]={K}\frac{\sum_{i=1}^{N}q_it_i}{f_{tot}}.
\end{equation}

For the upper bound of $\max _{i \in \mathcal{K}(\boldsymbol{q})^{(r)}}\left\{\tau_{i}\right\}$, we show that the probability of client $i \in \mathcal{K}(\boldsymbol{q})^{(r)}$ being the \emph{slowest} one  amongst the  sampled $K$ clients  is $\left(\sum\nolimits_{j=1}^{ i} q_{j}\right)^{\!K}\!\!-\!\left(\sum\nolimits_{j=1}^{i-1} \!q_{j}\right)^{\!K}$. %
Since we sample devices according to $\boldsymbol{q}$,
taking the expectation of all $N$ clients over time $\tau_i$ gives \eqref{Eperround_T_upper}.

Similarly, for the lower bound of $\min_{i \in \mathcal{K}^{(r)}(\boldsymbol{q})}\left\{\tau_{i}\right\}$, we show that the probability of client $i \in \mathcal{K}(\boldsymbol{q})^{(r)}$ being the \emph{fastest} one  amongst the  sampled $K$ clients  is $\left(\sum\nolimits_{j=i}^{N} q_{j}\right)^{K}-\left(\sum\nolimits_{j=i+1}^{N} q_{j}\right)^{K}$. Taking the expectation over of all $N$ clients with sampling probability $\boldsymbol{q}$ over time $\tau_i$ gives \eqref{Eperround_T_lower}.
\end{proof}}

{Although we can derive the analytical expression of the lower and upper bound of  $\Expect[T^{(r)}(\boldsymbol{q})]$, it is still difficult to use them to approximately solve the optimal sampling probability $\boldsymbol{q}$ in Problem \textbf{P1}. This is because the control variables $\boldsymbol{q}$  
are in a complex polynomial sum with an order $K$ in \eqref{Eperround_T_lower} and \eqref{Eperround_T_upper}, which prevents the access for the partial derivatives.}

{For analytical tractability, %
for any sampled clients $\mathcal{K}^{(r)}(\boldsymbol{q})$ in round  $r$, %
we define an average sampled clients' computation time %
$\bar{\tau}^{(r)}$ as 
\begin{equation}
\label{mean_appro_tau_bar_wireless}
\bar{\tau}^{(r)}:=\frac{\sum_{i=1}^{K}\tau_i}{K}, i \in \mathcal{K}^{(r)}(\boldsymbol{q}),
\end{equation}
where $\min_{i \in \mathcal{K}^{(r)}(\boldsymbol{q})}\left\{\tau_{i}\right\}<\bar{\tau}^{(r)}<\max_{i \in \mathcal{K}^{(r)}(\boldsymbol{q})}\left\{\tau_{i}\right\}$. 
Then, taking expectation over the sampled clients $\mathcal{K}(\boldsymbol{q})^{(r)}$ in \eqref{mean_appro_tau_bar_wireless}, %
we have 
\begin{equation}
\label{mean_appro_tau_bar}
\Expect_{\mathcal{K}(\boldsymbol{q})^{(r)}}\left[\bar{\tau}^{(r)}\right]:=\Expect_{i \in \mathcal{K}(\boldsymbol{q})^{(r)}}\left[\frac{\sum_{i=1}^{K}\tau_i}{K}\right]={{\sum_{i=1}^{N}q_i\tau_i}}.
\end{equation}
We can show that $\Expect\left[\min_{i \in \mathcal{K}^{(r)}(\boldsymbol{q})}\left\{\tau_{i}\right\}\right]<\Expect_{\mathcal{K}(\boldsymbol{q})^{(r)}}\left[\bar{\tau}^{(r)}\right]<\Expect\left[\max_{i \in \mathcal{K}^{(r)}(\boldsymbol{q})}\left\{\tau_{i}\right\}\right]$.
Therefore, based on \eqref{mean_appro_tau_bar},  we define the approximated expected round time of ${\Expect}\left[T^{(r)}(\boldsymbol{q})\right]$ in \eqref{E_T_total} as 
\begin{equation}
\label{mean_appro_Tr_bound_wireless}
\begin{aligned}
 \tilde{\Expect}\left[T^{(r)}(\boldsymbol{q})\right]&:=\frac{K\sum_{i=1}^{N}q_it_i}{f_{tot}}+\Expect_{\mathcal{K}(\boldsymbol{q})^{(r)}}\left[\bar{\tau}^{(r)}\right]\\
&=\frac{K\sum_{i=1}^{N}q_it_i}{f_{tot}}+{{\sum_{i=1}^{N}q_i\tau_i}}\\
&=\sum_{i=1}^{N}q_i\left(\frac{Kt_i}{f_{tot}}+\tau_i\right).
\end{aligned}
\end{equation}
The approximation $\tilde{\Expect}\left[T^{(r)}(\boldsymbol{q})\right]$ in \eqref{mean_appro_Tr_bound_wireless} equals to  ${\Expect}\left[T^{(r)}(\boldsymbol{q})\right]$ in \eqref{E_T_total} for the following two cases.

{\textbf{Case 1: homogeneous computation time $\tau_{i}$ }(i.e., $\tau_i=\tau_0, \forall i \in \mathcal{N}$)}

{For Case 1,  we rewrite  \eqref{total_bandwidth_wireless} as 
\begin{equation}
\label{case1_tau0}
\frac{\sum_{i=1}^{K}t_i}{T(\boldsymbol{q})^{(r)}-\tau_0}=f_{tot}, \ \forall i \in \mathcal{K(\boldsymbol{q})}^{(r)}.
\end{equation}
Then, taking the expectation of $T(\boldsymbol{q})^{(r)}$ in \eqref{case1_tau0}, we have
\begin{equation}
\label{case1_Tr_tau_0}
\begin{aligned}
 \Expect_{\mathcal{K}^{(r)}(\boldsymbol{q})}\left[T^{(r)}(\boldsymbol{q})\right]&=\Expect_{i \in \mathcal{K}^{(r)}(\boldsymbol{q})}\left[\frac{\sum_{i=1}^{K}t_i}{f_{tot}}\right]+\Expect_{i \in \mathcal{K}^{(r)}(\boldsymbol{q})}\left[\tau_0\right]\\
&={K}\frac{\sum_{i=1}^{N}q_it_i}{f_{tot}}+{{\sum_{i=1}^{N}q_i\tau_0}}\\
&=\sum_{i=1}^{N}q_i\left(\frac{Kt_i}{f_{tot}}+\tau_i\right)\\
&=\tilde{\Expect}_{\mathcal{K}(\boldsymbol{q})^{(r)}}\left[T^{(r)}(\boldsymbol{q})\right].
\end{aligned}
\end{equation}}

{Notably, for Case~1, the lower and upper bound of $\Expect_{\mathcal{K}(\boldsymbol{q})^{(r)}}\left[T^{(r)}(\boldsymbol{q})\right]$ in Theorem~\ref{lemma:expect_tr} are the same, since $\min_{i \in \mathcal{K}^{(r)}(\boldsymbol{q})}\left\{\tau_{i}\right\}=\max_{i \in \mathcal{K}^{(r)}(\boldsymbol{q})}\left\{\tau_{i}\right\}=\tau_0$.} 

\vspace{1mm}

\textbf{Case 2: heterogeneous $\tau_{i}$ when
 $K=1$.} 

\vspace{1mm}
 
{For Case~2, we  rewrite  \eqref{total_bandwidth_wireless} as  
 \begin{equation}
\label{case2_k1}
\frac{t_i}{T(\boldsymbol{q})^{(r)}-\tau_i}=f_{tot}, \ \forall i \in \mathcal{K(\boldsymbol{q})}^{(r)}.
\end{equation}
Then, taking the expectation of $T^{(r)}(\boldsymbol{q})$ in \eqref{case2_k1}, we have
\begin{equation}
\label{case2_Tr_k1}
\begin{aligned}
 \Expect_{\mathcal{K}^{(r)}(\boldsymbol{q})}\left[T^{(r)}(\boldsymbol{q})\right]&=\Expect_{i \in \mathcal{K}^{(r)}(\boldsymbol{q})}\left[\frac{t_i}{f_{tot}}\right]+\Expect_{i \in \mathcal{K}^{(r)}(\boldsymbol{q})}\left[\tau_i\right]\\
&=\frac{\sum_{i=1}^{N}q_it_i}{f_{tot}}+{{\sum_{i=1}^{N}q_i\tau_i}}\\
&=\sum_{i=1}^{N}q_i\left(\frac{t_i}{f_{tot}}+\tau_i\right)\\
&=\tilde{\Expect}_{\mathcal{K}^{(r)}(\boldsymbol{q})}\left[T^{(r)}(\boldsymbol{q})\right].
\end{aligned}
\end{equation}}

{Similarly, for Case~2, the lower and upper bound of $\Expect_{\mathcal{K}(\boldsymbol{q})^{(r)}}\left[T^{(r)}(\boldsymbol{q})\right]$ in Theorem~\ref{lemma:expect_tr} are the same, as we have $\Expect\left[\min_{i \in \mathcal{K}^{(r)}(\boldsymbol{q})}\left\{\tau_{i}\right\}\right]=\Expect\left[\max_{i \in \mathcal{K}^{(r)}(\boldsymbol{q})}\left\{\tau_{i}\right\}\right]=\sum_{i=1}^N \tau_i$ when $K=1$.}

For the general cases, we can still use   $\tilde{\Expect}[T^{(r)}(\boldsymbol{q})]$ from \eqref{mean_appro_Tr_bound_wireless}  to 
approximate $\Expect[T^{(r)}(\boldsymbol{q})]$. 
We summarize the insights of the obtained analytical round time  in \eqref{mean_appro_Tr_bound_wireless} as follows: %

\begin{itemize}
    \item 
 When we have unlimited or large enough bandwidth $f_{tot}$ (i.e., ${Kt_i}/{f_{tot}}\rightarrow0$), the approximated expected per-round time $\tilde{\Expect}[T^{(r)}(\boldsymbol{q})]$ in \eqref{mean_appro_Tr_bound_wireless} 
  is determined by the computation heterogeneity $\tau_i$, %
  which is independent to sampling number $K$ as clients' computation performs in parallel.
  \item In wireless networks, the communication is usually the bottleneck due to the limited bandwidth, e.g., the scale of communication time is larger than the scale of computation time. In particular, when $\tau_i\ll{Kt_i}/{f_{tot}}, \forall i \in \mathcal{N}$, the approximated expected round time $\tilde{\Expect}\left[T^{(r)}(\boldsymbol{q})\right]$  is roughly linearly increasing with sampling number $K$, as a larger $K$ reduces the communication bandwidth of each sampled client to upload the model parameters.  
\end{itemize} 
\textbf{Remark}: Although a smaller $K$ leads to shorter expected round time, it does not mean a smaller  $K$ can reduce the total convergence time. This is because the total FL convergence time also depends on the number of training rounds for reaching the target  precision. According to our obtained convergence result in Theorem~1, a smaller $K$ will leads to a slower convergence rate (e.g., a larger number of rounds  $R$ for reaching the same error $\epsilon$). Our experimental results in Section~\ref{sec:experimentation} also identify the impact of $K$ on the performance of convergence time.  %

\subsection{Approximate Optimization Problem for Problem \textbf{P1}}
Based on the approximated expected per-round learning time $\tilde{\Expect}[T^{(r)}(\boldsymbol{q})]$ in \eqref{mean_appro_Tr_bound_wireless} %
and by letting the analytical convergence upper bound in \eqref{convergence} satisfy the convergence constraint,\footnote{Optimization using upper bound as an approximation has also been adopted in \cite{wang2019adaptive,tran2019federated, chen2020convergence,luo2020cost}. %
} the original Problem \textbf{P1} can be approximated as
\begin{equation}
\begin{array}{cl}
\label{obj2}
\textbf{P2:} \  \min_{\boldsymbol{q}, R} &  \sum_{i=1}^{N}q_i\left(\dfrac{Kt_i}{f_{tot}}+\tau_i\right)R \\
\quad \text { s.t. } & \dfrac{1}{R}\left(\alpha {\sum_{i=1}^N\dfrac{p_i^2G_i^{2}}{Kq_i}} +\beta\right)\le \epsilon,\\
&\sum_{i=1}^Nq_i=1,\\
&  q_i>0,  \forall i \in \mathcal{N}, \ R \in \mathbb{Z}^{+}. 
\end{array}
\end{equation}%

Combining with \eqref{convergence}, we can see that Problem \textbf{P2} is more constrained than Problem \textbf{P1}, i.e., any feasible solution  of Problem \textbf{P2} is a subset of that of Problem~\textbf{P1}. 

We further relax $R$ as a continuous variable for theoretical analysis in Problem \textbf{P2}. For this relaxed  problem,  %
suppose ($\boldsymbol{q}^*$,  $R^*$) is the optimal solution, then we must have 
\begin{equation} \label{equality_hold}
\frac{1}{R^*}\left(\alpha {\sum_{i=1}^N\frac{p_i^2G_i^{2}}{Kq_i^*}} +\beta\right)= \epsilon.
\end{equation}
This is because if  \eqref{equality_hold} holds with an inequality, we could always find an $R^{\prime}< R^*$ that  
satisfies \eqref{equality_hold} with equality, but the solution ($q^*$, $R^{\prime}$) can further reduce the objective function value. Therefore, for the optimal $R$, \eqref{equality_hold} always holds, and we can obtain $R$ from \eqref{equality_hold}
and substitute it %
into the objective of Problem \textbf{P2}. Then, the objective of Problem \textbf{P2} is
\begin{equation}
\label{obj22}
\left[\sum_{i=1}^{N}q_i\left(\frac{Kt_i}{f_{tot}}+\tau_i\right)\right]\left(\alpha {\sum\limits_{i=1}^N\frac{p_i^2G_i^{2}}{Kq_i}} +\beta\right),
\end{equation}
which\footnote{For ease of analysis, we omit  $\epsilon$ as it is a constant multiplied by the entire objective function.} %
is only associated with client sampling probability $\boldsymbol{q}$. In line with this,   
\textbf{P2} can be expressed as  
\begin{equation}
\label{ob3}
\begin{array}{cl}
\textbf{P3:} \  
\min_{\boldsymbol{q}}  & \left[\sum_{i=1}^{N}q_i\left(\dfrac{Kt_i}{f_{tot}}+\tau_i\right)\right]\left(\alpha {\sum_{i=1}^N\dfrac{p_i^2G_i^{2}}{Kq_i}} +\beta\right) \\
\ \ \text {s.t.} & \!\!\!\! \sum_{i=1}^Nq_i=1,\  q_i>0,  \forall i \in \{1,\ldots N\}. 
\end{array}
\end{equation}

\textbf{Remark:} %
The objective function of Problem \textbf{P3} is in a more straightforward form compared to Problem \textbf{P2}. However, to solve for the optimal sampling probability $\boldsymbol{q}$, we need to know the value of the parameters in \eqref{ob3}, e.g., $G_i$, $\alpha$, and $\beta$, which is challenging because we can only estimate them during the training process of FL.\footnote{We assume that clients' heterogeneous computation time $\tau_i$,  communication time $t_i$, and their dataset size proportion $p_i$ can be measured offline.} 

In the following, we %
solve Problem \textbf{P3} as an approximation of the original Problem \textbf{P1}. 
Our empirical results in Section~\ref{sec:experimentation} demonstrate that the solution obtained from solving Problem \textbf{P3} %
achieves the target precision with less  time compared to baseline client sampling schemes. %

\subsection{Solving Problem  \textbf{P3}}
In this subsection, we first show how to obtain the unknown parameters, e.g., $G_i$, $\alpha$ and $\beta$. %
Then, we develop an efficient algorithm to  solve Problem \textbf{P3}. We summarize the overall algorithm in  Algorithm~\ref{opt_algorithm}. 
Finally, we identify some insightful solution properties. 

\subsubsection{\textbf{Estimation of Parameters  $G_i$ and $\frac{\beta}{\alpha }$}}
We first show how to estimate $\frac{\beta}{\alpha}$ via a substitute sampling scheme.\footnote{We only need to estimate the value of $\frac{\beta}{\alpha }$ instead of $\alpha$ and $\beta$ each, because we can divide parameter $\alpha$ on the objective of Problem \textbf{P3} without affecting the optimal sapling solution.}  Then, we show that we could indirectly acquire the knowledge of $G_i$ during the estimation process of $\frac{\beta}{\alpha}$.

The basic idea is to utilize the derived convergence upper bound in \eqref{convergence} to  approximately solve $\frac{\beta}{\alpha}$ as a variable,  
via performing Algorithm~1 with two
baseline sampling schemes: uniform sampling $\mathbf{q_1}$ with ${q}_i=\frac{1}{N}$ and  weighted sampling $\mathbf{q_2}$ with ${q}_i=p_i$, respectively. 

{Note that for sampling schemes $\mathbf{q_1}$ and $\mathbf{q_2}$, we only run a limited number of rounds to reach a predefined loss $F_s$. We don't run these schemes until they converge to the precision $\epsilon$.}
This is because our goal is to find and run with the optimal sampling scheme $\mathbf{q^*}$ so that we can achieve the target precision with the minimum wall-clock time. %

Specifically, suppose $R_{\mathbf{q_1},s}$ and $R_{\mathbf{q_2},s}$ are the number of rounds for reaching the pre-defined loss $F_s$ for schemes $\mathbf{q_1}$ and $\mathbf{q_2}$, respectively.  According to \eqref{convergence}, we have 
\begin{equation}
\begin{cases}
    \label{A0B01}
   \left(F_s-F^{*}\right)R_{\mathbf{q_1},s} \approx \alpha N{\sum_{i=1}^Np_i^2G_i^{2}}/K +\beta,\\
      \left(F_s-F^{*}\right)R_{\mathbf{q_2},s} \approx \alpha {\sum_{i=1}^N{p_iG_i^{2}}}/K +\beta.
 \end{cases}
\end{equation}
Based on \eqref{A0B01}, we have
  \begin{equation}
    \label{A0B02}
    \frac{R_{\mathbf{q_1},s}}{R_{\mathbf{q_2},s}} \approx 
    \frac{ \alpha N{\sum_{i=1}^Np_i^2G_i^{2}}/K +\beta}{\alpha {\sum_{i=1}^N{p_iG_i^{2}}}/K +\beta}.
\end{equation}
We can calculate $\frac{\alpha }{\beta}$ from \eqref{A0B02} once we know the value of  $G_i$. 

Notably, we can estimate $G_i$ during the procedure of estimating $\frac{\alpha }{\beta}$. The idea is to let the sampled clients send back the norm of their local SGD along with their returned local model updates, and then the server updates $G_i$ with the received norm values. This approach does not add much communication overhead, since we only need to additionally transmit the value of the gradient norm (e.g., only a few bits for quantization) instead of the full gradient information.

In practice, %
due to the sampling variance, we may set several different $F_s$  %
to obtain an averaged estimation of $\frac{\alpha }{\beta}$. %
The overall estimation process corresponds to Lines~\ref{alg:optimalSolution:startEstimation}--\ref{alg:optimalSolution:endEstimation} of Algorithm~\ref{alg:optimalSolution}.

\begin{algorithm}[t]
\label{opt_algorithm}
\small
\caption{Approximate Optimal Client Sampling  for Federated Learning with System and Statistical Heterogeneity}
\label{alg:optimalSolution}
	\KwIn{{$N$, $K$, $E$, $\tau_i$, $t_i$, $f_{\text{tot}}$, $p_i$,  $\mathbf{w}_0$}, loss  $F_s$,  step-size  $\epsilon_0$}
	\KwOut{{Approximation of $\boldsymbol{q}^*$}}

\For{$s=1,2, \ldots, S$  \label{alg:optimalSolution:startEstimation}}{
Server runs Algorithm 1 with uniform sampling $\mathbf{q_1}$ and weighted sampling $\mathbf{q_2}$, respectively;

The sampled clients send back their local gradient norm information along with their updated models;  

Server updates all clients' $G_i$ based on the received gradient norms;

Server record $R_{\mathbf{q_1},s}$ and $R_{\mathbf{q_2},s}$ when reaching $F_s$;
}

Calculate average $\frac{\beta}{\alpha }$ using \eqref{A0B02}; \label{alg:optimalSolution:endEstimation}

\For{$M(\epsilon_0)=M_{min}, M_{min}+\epsilon_0, M_{min}+2\epsilon_0 \ldots, M_{max}$  \label{alg:optimalSolution:startSolve}}{ %
Substitute $M(\epsilon_0)$, $\frac{\beta}{\alpha }$, $N$,  $t_i$, $p_i$, $G_i$ into \textbf{P4}; 

Solve \textbf{P4} via CVX, and obtain $\boldsymbol{q}^*(M(\epsilon_0))$ %
}

\Return $\boldsymbol{q}^*(M^*(\epsilon_0))=\arg\min_{M(\epsilon_0)}\boldsymbol{q}^*(M(\epsilon_0))$ %
\label{alg:optimalSolution:endOptimization}

\end{algorithm}

\subsubsection{\textbf{Optimization Algorithm for  $\boldsymbol{q}^*$}} %
{We first identify the property of Problem \textbf{P3} %
and then show how to compute $\boldsymbol{q}^*$. 
\begin{lemma}
\label{theorem:biconvex}
Problem \textbf{P3} %
is non-convex. 
\end{lemma}
\begin{proof}
We prove this lemma by showing that the Hessian of the objective function in Problem \textbf{P3} is not positive semi-definite. 
{E.g., for $N=2$ case, we have
$\frac{\partial^2 \tilde{\Expect}[T_\textnormal{tot}]}{\partial^2 q_1}=\frac{2\alpha q_2\left(\frac{Kt_2}{f_{tot}}+\tau_2\right)p_1^2G_1^2}{q_1^3}>0$, %
whereas
$\frac{\partial^2 \tilde{\Expect}[T_\textnormal{tot}]}{\partial^2 q_1}\frac{\partial^2 \tilde{\Expect}[T_\textnormal{tot}]}{\partial^2 q_2}\!-\!\left(\frac{\partial^2 \tilde{\Expect}[T_\textnormal{tot}]}{\partial q_1\partial q_2}\right)^{2}\!=\!-\alpha ^2\left(\frac{\left(\frac{Kt_1}{f_{tot}}+\tau_1\right)p_2^2G_2^2}{q_2^2}-\frac{\left(\frac{Kt_2}{f_{tot}}+\tau_2\right)p_1^2G_1^2}{q_1^2}\right)^{2}\!\le \!0$}, which concludes the proof that Problem \textbf{P3} is non-convex.
\end{proof}}

{To efficiently solve Problem \textbf{P3}, we  define a new control variable %
\begin{equation}
    \label{new_M}
    M:=\sum\nolimits_{i=1}^Nq_i\left(\frac{Kt_i}{f_{tot}}+\tau_i\right),
\end{equation}
where $M_{min}=\min_{i\in \mathcal{N}} \left(\frac{Kt_i}{f_{tot}}+\tau_i\right) \le M \le M_{max}=\max_{i\in \mathcal{N}} \left(\frac{Kt_i}{f_{tot}}+\tau_i\right)$.
Then, we rewrite Problem \textbf{P3} as 
\begin{equation}
\label{ob4}
\begin{array}{cl}
 \!\!\!\! \!\!\textbf{P4:} \ \
 \min_{\boldsymbol{q},M} &g(\boldsymbol{q},M) = M \cdot \left(\alpha {\sum_{i=1}^N\frac{p_i^2G_i^2}{Kq_i}} +\beta\right) \\
\ \ \text {s.t.} &  \sum_{i=1}^Nq_i=1, \\ &\sum_{i=1}^Nq_it_i=M,\\
&  q_i>0,  \forall i \in \mathcal{N}.
\end{array}
\end{equation}
For any fixed feasible value of $M \in [M_{min}, M_{max}]$, Problem \textbf{P4} is convex in $\boldsymbol{q}$, because the objective function is strictly convex and the constraints are linear.} 

{We will solve Problem \textbf{P4} in two steps. First, for any fixed $M$, we will solve the optimal  $\boldsymbol{q}^*(M)$ in Problem \textbf{P4} %
via a convex optimization tool, e.g., CVX.  This allows us to write the objective function of Problem P4 as $g(\boldsymbol{q}^\ast(M), M)$. Then we will  solve the problem by using 
a linear search method over $M$ with a fixed step-size $\epsilon_0$ over the interval $[M_{min}, M_{max}]$.\footnote{We denote $M^*(\epsilon_0)$ and the corresponding $\boldsymbol{q}^*(M^*(\epsilon_0))$ in the searching domain to approximate the optimal $M^*$ and  $\boldsymbol{q}^*$ in Problem \textbf{P4}.}  %
This optimization process
corresponds to Lines \ref{alg:optimalSolution:startSolve}--\ref{alg:optimalSolution:endOptimization} of Algorithm~2.} %

\textbf{Remark}: Our optimization algorithm is efficient in the sense that the linear search method domain $[M_{min}, M_{max}]$ is independent of the scale of the problem, e.g., number of $N$. %

\subsubsection{\textbf{Property of $\boldsymbol{q}^*$}}
{Next we %
show some interesting properties of the optimal sampling strategy. %
\begin{theorem}
\label{property1}
Consider the optimal solution of %
Problem \textbf{P3} $\boldsymbol{q}^*$. For any two different clients $i$ and $j$ with $\frac{Kt_i}{f_{tot}}+\tau_i\le \frac{Kt_j}{f_{tot}}+\tau_j$ and $p_iG_i\ge p_jG_j$, we have $q_i^*\ge q_j^*$. 
\end{theorem}
\begin{proof}
We prove this theorem by contradiction. Suppose  $q_i^*<q_j^*$, when $T_i\le T_j$ and $p_iG_i\ge p_jG_j$. Then, we can simply let $q_i^\prime=q_j^*, q_j^\prime=q_i^*$ such that $q_i^\prime>q_j^\prime$ and achieve a smaller $\Expect[T_\textnormal{tot}(q_i^\prime,q_j^\prime)]$ than   $\Expect[T_\textnormal{tot}(q_i^*,q_j^*)]$, which concludes the proof.
\end{proof}}

Theorem~\ref{property1} shows that the optimal client sampling strategy %
would allocate higher probabilities to those who have smaller values $\frac{Kt_i}{f_{tot}}+\tau_i$ (e.g., having both powerful computation capability and good communication channel) and larger product value of $p_iG_i$. This 
characterizes the impact and interplay between system heterogeneity and statistical heterogeneity.

Although it is infeasible to derive an   %
analytical  solution of $\boldsymbol{q}^*$ regarding the  impact of $\tau_i$, $t_i$, and  $p_iG_i$ on $\boldsymbol{q}^*$,   %
we show that when  ${\beta}/{\alpha }\rightarrow0$ holds for a certain FL task, 
we have 
the global optimal solution $\boldsymbol{q}^*$ of Problem \textbf{P3} as
\begin{equation}
    \label{optsampling_statis_sys}
q_{i}^{*}={\left(\frac{p_{i} G_{i}}{\sqrt{\frac{Kt_i}{f_{tot}}+\tau_i}}\right)}\Big/{\left(\sum_{j=1}^{N} \frac{p_{j} G_{j}}{\sqrt{\frac{Kt_j}{f_{tot}}+\tau_j}}\right)}.
\end{equation}

This is because when ${\beta}/{\alpha }\rightarrow0$, we must have $\left(\sum\nolimits_{i=1}^Nq_it_i\right)\frac{\beta}{\alpha}\rightarrow 0$ given that $\sum\nolimits_{i=1}^Nq_i\left(\frac{Kt_i}{f_{tot}}+\tau_i\right)$ is bounded by $[M_{min}, M_{max}]$. Therefore, the objective of Problem \textbf{P3} can be rewritten as
\begin{equation}
\label{special_case}
  \min_{\boldsymbol{q}}   \left[\sum\nolimits_{i=1}^Nq_i\left(\frac{Kt_i}{f_{tot}}+\tau_i\right)\right]\left( {\sum\nolimits_{i=1}^N\frac{p_i^2G_i^{2}}{q_i}} \right).  
\end{equation}
By Cauchy-Schwarz inequality, we have %
\begin{equation}
    \label{special_case_prove}
    \begin{aligned}
   \left(\sum\limits_{i=1}^N\left(\sqrt{q_i\left(\frac{Kt_i}{f_{tot}}+\tau_i\right)}\right)^2\right)\left({\sum\limits_{i=1}^N\left(\frac{p_iG_i}{\sqrt{q_i}}\right)^{2}} \right)\\ %
   \ge
   \left(\sum\limits_{i=1}^N\sqrt{\left(\frac{Kt_i}{f_{tot}}+\tau_i\right)}\cdot{p_iG_i}\right)^{2}.
         \end{aligned}
\end{equation}
Hence, the minimum of \eqref{special_case} is $\left(\sum_{i=1}^N\sqrt{t_i}\cdot{p_iG_i}\right)^{2}$, which  %
is independent of sampling strategy  $\boldsymbol{q}$. The equality of \eqref{special_case_prove}  holds, meaning that the global optimal solution is achieved, if and only if only when  \eqref{optsampling_statis_sys} holds.

Though $\boldsymbol{q}^*$ in \eqref{optsampling_statis_sys} is valid only for the special case of $\frac{\beta}{\alpha }\rightarrow0$, the global optimal sampling solution in \eqref{optsampling_statis_sys} characterizes an analytical interplay between the system heterogeneity (i.e., $\tau_i$ and ${t_i}$) and statistical heterogeneity (i.e., $p_iG_i$).

\section{Experimental Evaluation}
\label{sec:experimentation}
In this section, we empirically evaluate the performance of our proposed  client sampling scheme. We first present  the evaluation setup and then show the experimental results. {Our experiment code is available at: \\
\url{https://github.com/WENLIXIAO-CS/WirelessFL}}

\begin{table*}[!t]
 \caption{{Estimation of $\frac{\alpha}{\beta}$  for three Setups: value of different estimation loss $F_s$ and the corresponding number of rounds for reaching  $F_s$ using uniform sampling and weighted sampling}}
  \centering
  {
\begin{tabular}{c||c||c|c|c|c|c||c}
\toprule[1.2pt]
\multirow{3}{*}{\makecell[c]{\textbf{Setup 1}\\(EMNIST dataset)}}%
    & {Estimation loss $F_s$} & $1.7$ &  $1.6$ & $1.5$ & $1.4$ &$1.3$  &  \multirow{3}{*}{{\makecell[c]{\textbf{Estimated}\\$\dfrac{\alpha}{\beta}=11.51$}}  } \\ \cline{2-7}
                       &                       {Rounds for reaching $F_s$ with uniform sampling} &  $39$ & $47$ & $58$ & $79$ & $132$ &                \\ \cline{2-7}
                       &                        {Rounds for reaching $F_s$ with  weighted sampling} &   $12$ & $19$ & $27$ & $36$ & $56$  &                 \\ \midrule[1.2pt]

\multirow{3}{*}{\makecell[c]{\textbf{Setup 2}\\(Synthetic dataset)}}        & {Estimation Loss $F_s$}&    $1.2$&    $1.15$ &  $1.1$ & $1.05$ & $1.0$  &    \multirow{3}{*}{{\makecell[c]{\textbf{Estimated}\\$\dfrac{\alpha}{\beta}=63.88$}}  } \\ \cline{2-7}
                                                  & Rounds for reaching $F_s$ with uniform sampling  &  $71$&$81$&$95$& $109$ &$133$  &                \\ \cline{2-7}
                       &                        Rounds for reaching $F_s$ with  weighted sampling  & $48$& $57$ & $69$ & $84$& $101$ &               \\ \bottomrule[1.2pt]
\multirow{3}{*}{\makecell[c]{\textbf{Setup 3}\\ (MNIST dataset)}}       & {Estimation Loss $F_s$}&    $0.3$&    $0.2875$ &  $0.275$ & $0.2625$ & $0.25$  &    \multirow{3}{*}{{\makecell[c]{\textbf{Estimated}\\$\dfrac{\alpha}{\beta}=4.92$}}} \\ \cline{2-7}
                                                  & Rounds for reaching $F_s$ with uniform sampling  &  $46$&$48$&$51$& $55$ &$59$  &                \\ \cline{2-7}
                                              & Rounds for reaching $F_s$ with  weighted sampling  & $27$& $28$ & $30$ & $32$& $34$ &               \\ \bottomrule[1.2pt]
\end{tabular}}
 \label{estimation_process}
 \end{table*}

\subsection{Experimental Setup}
\subsubsection{Platforms}
We conduct experiments both on a networked hardware prototype system and in a simulated environment.\footnote{The prototype implementation allows us to capture real system operation time, and the simulation system allows us to simulate large-scale FL environments with manipulative parameters.} 
As illustrated in Fig. \ref{HDP_new}, our prototype system consists of $N=40$ Raspberry Pis serving as clients %
and a laptop computer acting as the central server. All devices are interconnected via an enterprise-grade Wi-Fi router. We develop a TCP-based socket interface for the communication between the server and clients with bandwidth control. In the simulation system, we simulate $N=100$ virtual devices and a virtual central server.  

\subsubsection{Datasets and Models} We evaluate our results on two real datasets and a synthetic dataset. For the real dataset, we adopt the widely used MNIST dataset and EMNIST dataset %
\cite{li2018federated}. %
For the synthetic dataset, we follow a similar setup to that in \cite{li2019convergence}, which generates $60$-dimensional random vectors as input data. %
We adopt both the \emph{convex} {multinomial logistic regression} model and the \emph{non-convex}  convolutional
neural network (CNN) model with LeNet-5 architecture \cite{lecun1998gradient}.%

\begin{figure}[!t]
	\centering
	\includegraphics[width=8.8cm,height=4.6cm]{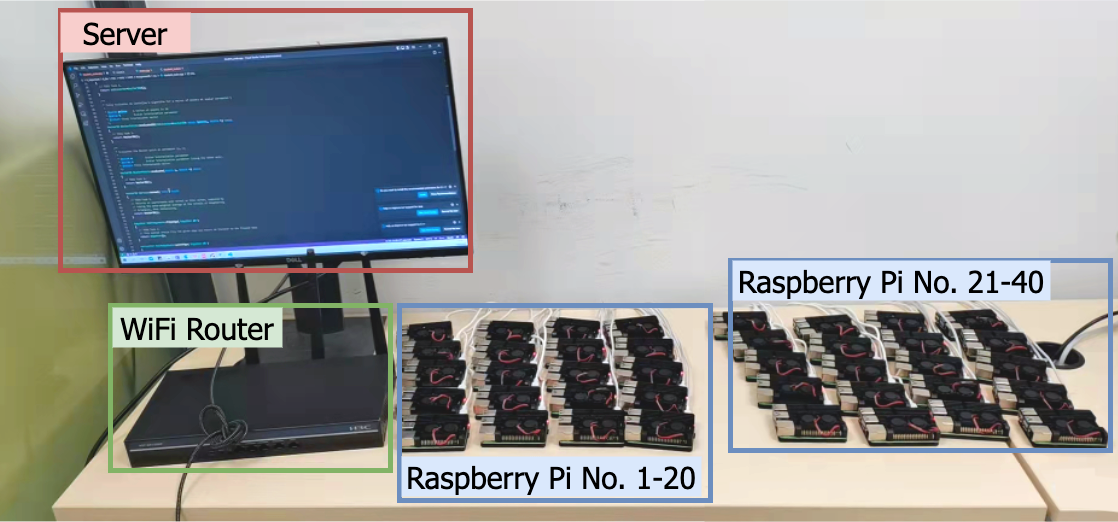}
	\caption{Hardware prototype with the laptop serving as the central server and 40 Raspberry Pis %
	serving as clients.} %
\label{HDP_new}
\end{figure}

\subsubsection{Training  Parameters}  For all experiments, we initialize our model with $\mathbf{w}_0\!=\!\mathbf{0}$ and use a SGD batch size $b=24$. We use an initial learning rate $\eta_0 =0.1$ with a decay rate $\frac{{\eta}_0}{1+r}$, where $r$ is the communication round index. We adopt the similar FedAvg settings as in \cite{bonawitz2019towards,nishio2019client,chen2020optimal,li2019convergence}, where we sample $10$\% of all clients in each round, i.e., $K\!=\!4$ for Prototype Setup and $K\!=\!10$ for Simulation Setups, with each client performing $E=50$ local iterations.\footnote{We also conduct experiments both on Prototype and Simulation Setups with variant $E$ and $K$,  which show a similar performance as the experiments in this paper.} %

\subsubsection{Heterogeneous System Parameters}
For the Prototype Setup, the  computation time for $E=50$ local iterations is almost the same amongst 40 devices, which we measured as $\tau_i\approx 0.5$ second.  %
To enable a more heterogeneous communication time, we control clients'  bandwidth and generate a uniform distribution $t_i/f_{tot} \sim \mathcal{U}(0.22,5.04)$ second, with a mean of $2.63$ second and the standard deviation of $1.39$ second.
For the simulation system, we generate  clients' computation and communication time  based on  an exponential distribution, i.e., $\tau_i \sim \exp{(1)}$ second and $t_{i}/f_{tot} \sim \exp{(1)}$ second with  
both mean and standard deviation as $1$ second.%

\subsubsection{Implementation}
we consider three experimental setups.
\begin{itemize}
    \item \textbf{Setup 1}: We conduct the first experiment on the prototype system using  logistic regression and the EMNIST dataset. To generate heterogeneous data partition, similar to \cite{li2018federated}, we randomly subsample $33,036$ lower case character samples from the EMNIST dataset and distribute  among $N\!=\!40$ edge devices in an \emph{unbalanced} (following the  power-law distribution) and \emph{non-i.i.d.}  (i.e., each device has 1--10 classes) fashion.\footnote{The number of samples and the number of classes are randomly matched, such that clients with more data samples may not have more classes.} %
    \item   \textbf{Setup 2}: We conduct the second experiment in the simulated system using  logistic regression and the Synthetic dataset. To simulate a heterogeneous setting, we use the non-i.i.d.  $Synthetic \ (1, 1)$ setting. We generate $20,509$ data samples and distribute them among $N\!=\!100$ clients in an \emph{unbalanced} power-law distribution.%

     \item \textbf{Setup 3}: We conduct the third experiment in the simulated
system using \emph{non-convex CNN} and the MNIST dataset, where we randomly subsample 
$15,129$ data samples from the MNIST dataset and distribute them among $N \!=\! 100$ clients in an \emph{unbalanced} (following the power-law distribution) and \emph{non-i.i.d.}  (i.e., each device has 1--6 classes) fashion.%
     
\end{itemize}
We present the estimation process and results of $\alpha/\beta$ 
for the three experiment setups in Table~\ref{estimation_process}.

\begin{table*}[t!]  \label{all_figure_summary}
{ \caption{{Wall-clock Time for  Reaching Target {Loss} for Different Sampling Schemes}}}
  \centering
  \begin{threeparttable}  
    \begin{tabular}{c||c|c|c|c}
    \toprule
  \diagbox[]{Setups}
   {Schemes} &  \textbf{proposed sampling} & statistical sampling & weighted sampling & uniform sampling  \bigstrut\\
    \hline
   {\makecell[c]{\textbf{Setup 1}\\(EMNIST dataset)}} &  {$\mathbf{2428.7 \pm 707.3}$ \textbf{s}} & {$3958.3 \pm 302.4$ s ($\mathbf{1.6\times}$)} & {$3477.4 \pm 334.2$ s ($\mathbf{1.4\times}$)} & {$8487.4 \pm 528.8$ s ($\mathbf{3.5\times}$)}\tnote{\dag}  \bigstrut\\
    \hline
   {\makecell[c]{\textbf{Setup 2} \\ (Synthetic dataset)} }  & {$\mathbf{9767.8 \pm 1454.3}$ \textbf{s}} & {$25730.0 \pm 2358.7$ s ($\mathbf{2.6\times}$)}& {$24207.9 \pm 1884.9$ s ($\mathbf{2.5\times}$)} & {$17514.3 \pm 453.8$ s ($\mathbf{1.8\times}$)}  \bigstrut\\
    \hline
{\makecell[c]{\textbf{Setup 3} \\(MNIST dataset)}} &  {$\mathbf{2453.7 \pm 1289.2}$ \textbf{s}} & {$3295.0 \pm 878.3$ s ($\mathbf{1.3\times}$)} & {$4548.6 \pm 1530.2$ s ($\mathbf{1.9\times}$)} & NA \bigstrut\\
    \bottomrule
    \end{tabular}%
     \begin{tablenotes}    
        \small%
        \item \dag \  ``$3.5\times$"  represents the wall-clock time ratio of uniform sampling over proposed sampling  for reaching the target loss, which is equivalent to proposed sampling takes $71$\% less time than uniform sampling.      
      \end{tablenotes}            
    \end{threeparttable}    
\end{table*}

\begin{figure*}[!t]
\centering
\subfigure[Loss with wall-clock time]{\label{hd_lossa}\includegraphics[width=5.95cm]{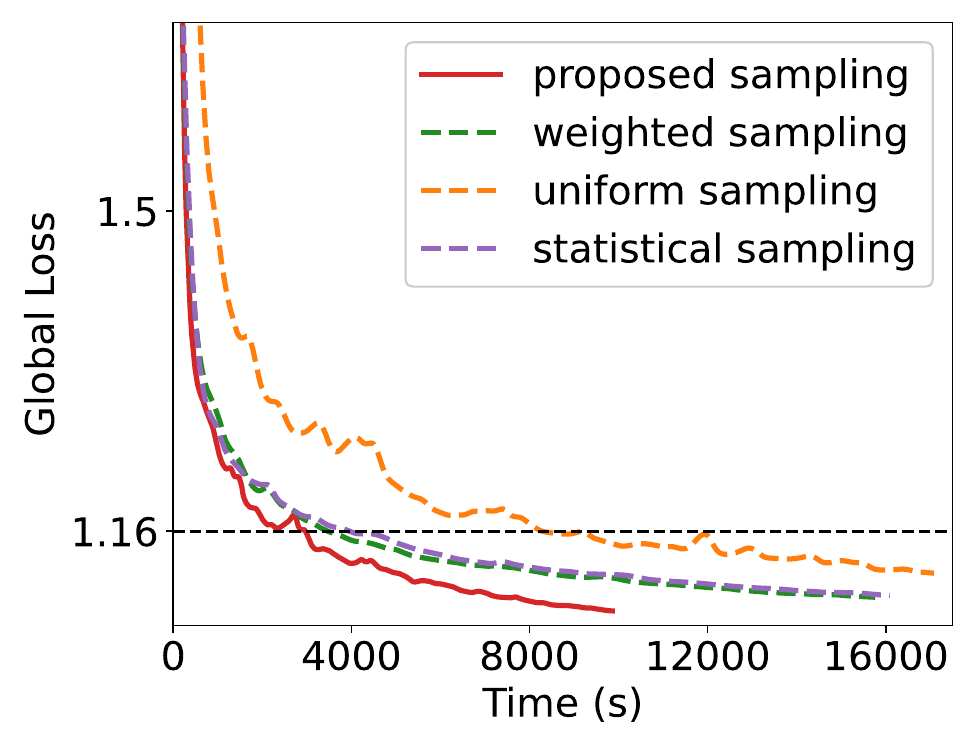}}
\subfigure[Accuracy with wall-clock time]{\label{hd_lossb}\includegraphics[width=5.95cm]{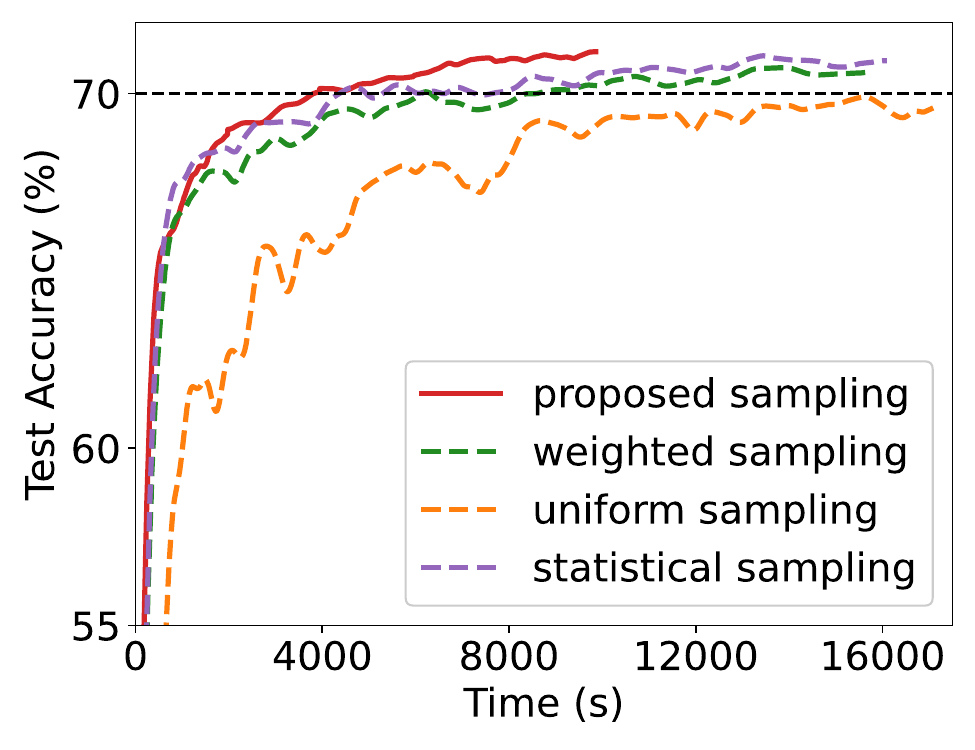}}
\subfigure[Loss with number of rounds]{\label{hd_lossc}\includegraphics[width=5.95cm]{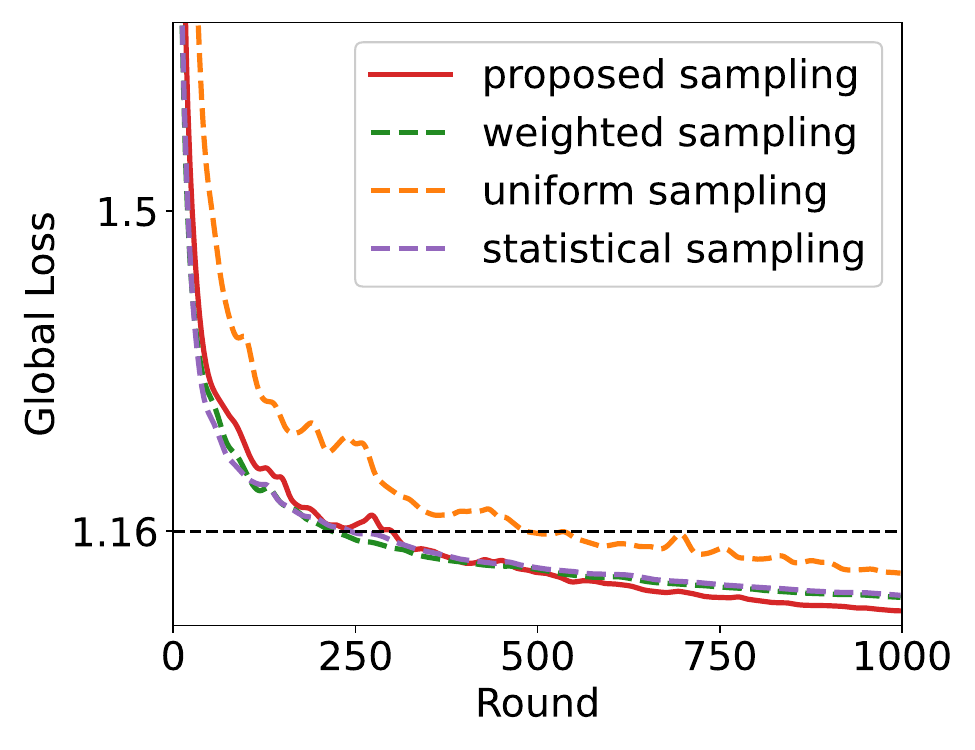}}
\caption{%
{Performances of \textbf{Setup 1} with logistic regression and EMNIST dataset for reaching  target loss $1.16$ and target accuracy $70$\%. %
}}
\label{hd}
\end{figure*}

\begin{figure*}[ht]
\centering
\subfigure[Loss with wall-clock time]{\label{soft_loss}\includegraphics[width=5.95cm]{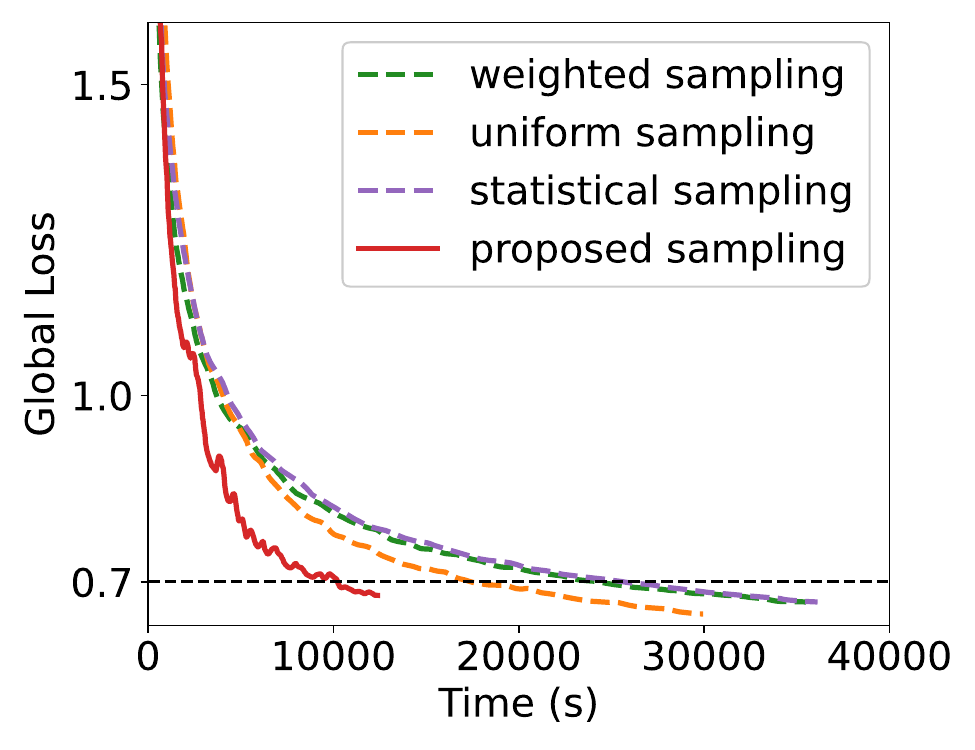}}
\subfigure[Accuracy with wall-clock time]{\label{soft_acc}\includegraphics[width=5.95cm]{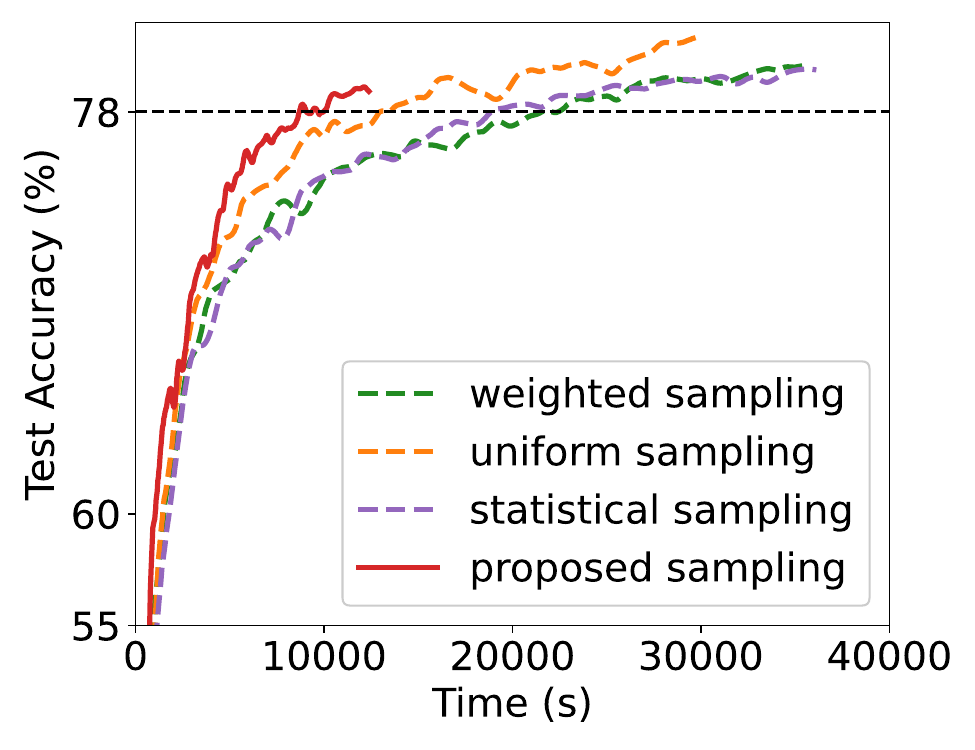}}
\subfigure[Loss with number of rounds]{\label{soft_loss_round}\includegraphics[width=5.95cm]{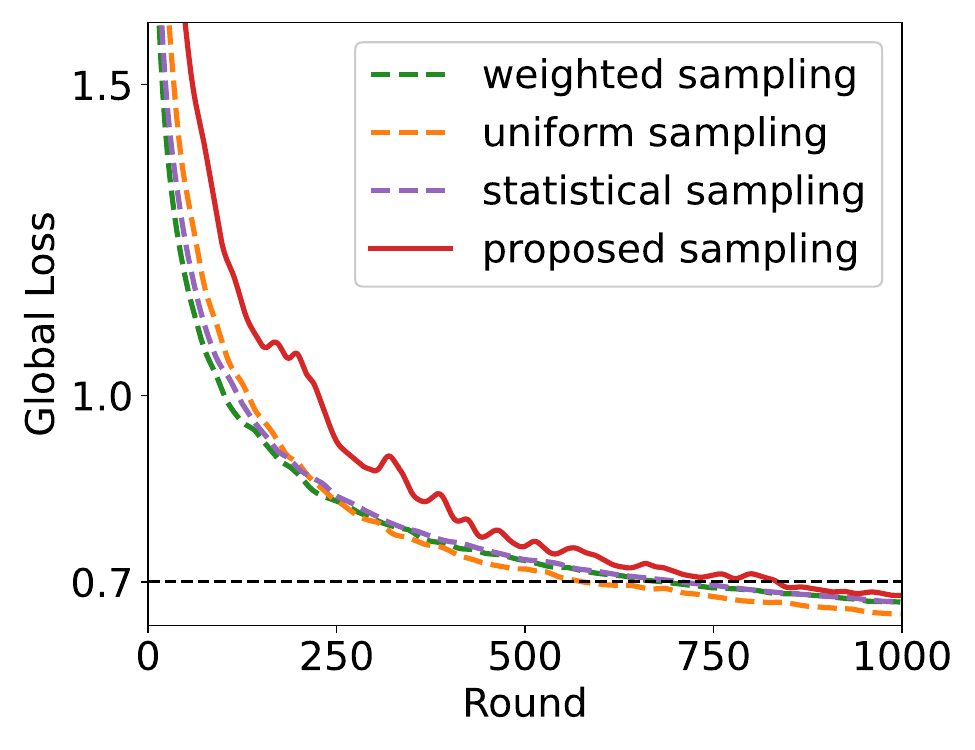}}
\caption{%
{Performances of \textbf{Setup 2} with logistic regression  and $Synthetic$ dataset for reaching    %
target loss $0.7$ and target accuracy $78$\%. %
}}
\label{soft}
\end{figure*}

\subsection{Experimental Results}
In this section, we first compare the performance of our proposed  client sampling scheme (Algorithm 2) with different sampling baselines. Then, we evaluate the impact of the key system parameter on the total convergence time. We evaluate the wall-clock time performances of both the  global training loss and test accuracy on the aggregated model in each round for all sampling schemes. 
We average each experiment over $20$ independent runs. For a fair comparison, we use the same random seed to compare sampling schemes in a single run and vary random seeds across different runs.

\subsubsection{Comparison with Different Sampling Schemes}
We compare the performance of our proposed  client sampling scheme  with three benchmarks 
 1)  \emph{uniform sampling},
  2) \emph{weighted sampling}, 
  and 3) \emph{statistical sampling} where we sample clients according to statistical heterogeneity without system heterogeneity.  
Benchmarks 1--2 are widely adopted for convergence guarantees in \cite{li2019convergence,haddadpour2019convergence,karimireddy2019scaffold,yang2021achieving, qu2020federated}. The third baseline is  an offline variant of the proposed schemes in \cite{chen2020optimal, rizk2020federated}.\footnote{The client sampling in \cite{chen2020optimal, rizk2020federated} is weighted by the norm of the local stochastic gradient in each round, which frequently requires the knowledge of stochastic gradient from all clients to calculate the sampling probabilities.}
Fig.~\ref{hd}, Fig.~\ref{soft}, and Fig.~\ref{soft2} show the results of Setup 1,  Setup 2, and Setup 3, respectively. We summarize the key observations as follows.

\begin{itemize}
    \item 
{\emph{Loss with Wall-clock Time:}
As predicted by our theory, Fig.~\ref{hd}(a), Fig.~\ref{soft}(a), and Fig.~\ref{soft2}(a) show that \emph{our proposed sampling scheme achieves the same target loss with significantly less time}, compared to the rest of baseline sampling schemes. Specifically, for Prototype Setup $1$ in Fig.~\ref{hd}(a), our proposed sampling scheme spends around $71$\% less time  %
than uniform sampling,  $39$\% less time than statistical sampling, and $30$\% less time than  weighted sampling for  reaching  the  same  target  loss. Fig.~\ref{soft2}(a) highlights the fact that  our
proposed sampling works well with the \emph{non-convex CNN} model, under which the  naive uniform sampling cannot reach the target loss within $4,000$ seconds, indicating the importance of a careful client sampling design.   
Table~$3$ summarizes the superior performances of our proposed sampling scheme in wall-clock time for reaching target {loss} in all three setups.} %

\begin{figure*}[!t]
\centering
\subfigure[Loss with wall-clock time]{\label{soft_loss_2}\includegraphics[width=5.95cm]{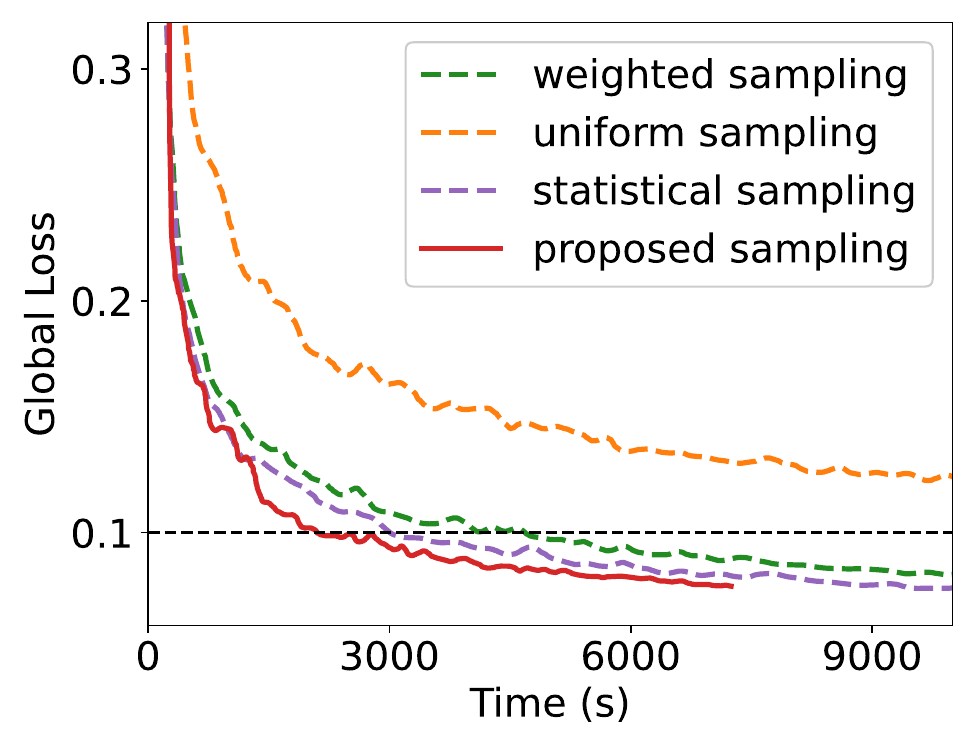}}
\subfigure[Accuracy with wall-clock time]{\label{soft_acc_2}\includegraphics[width=5.95cm]{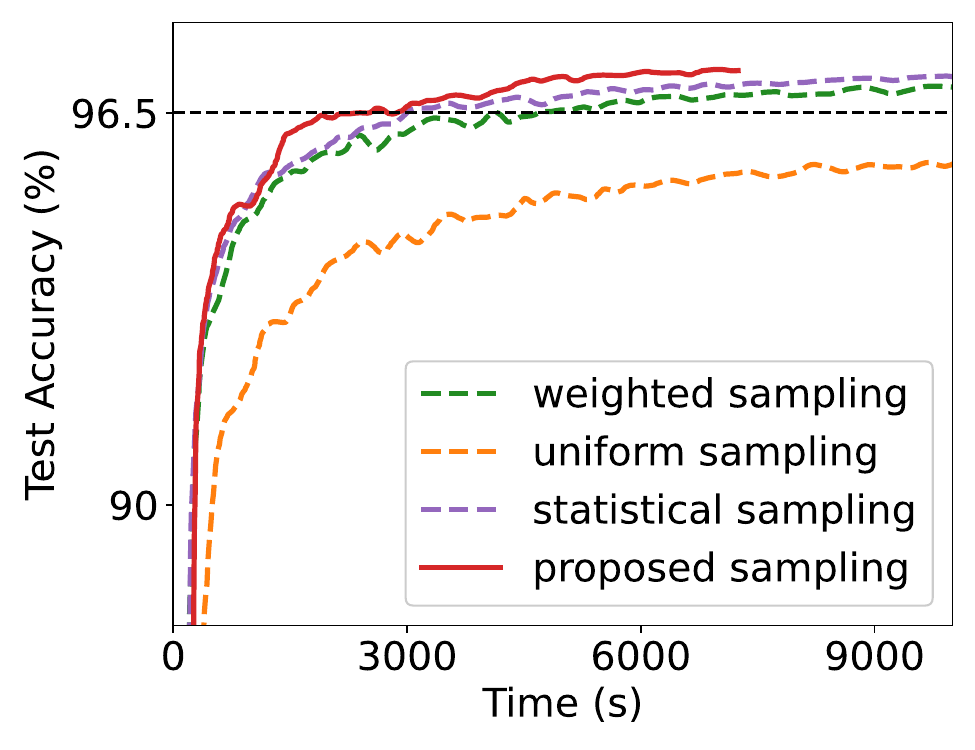}}
\subfigure[Loss with number of rounds]{\label{soft_loss_round_2}\includegraphics[width=5.95cm]{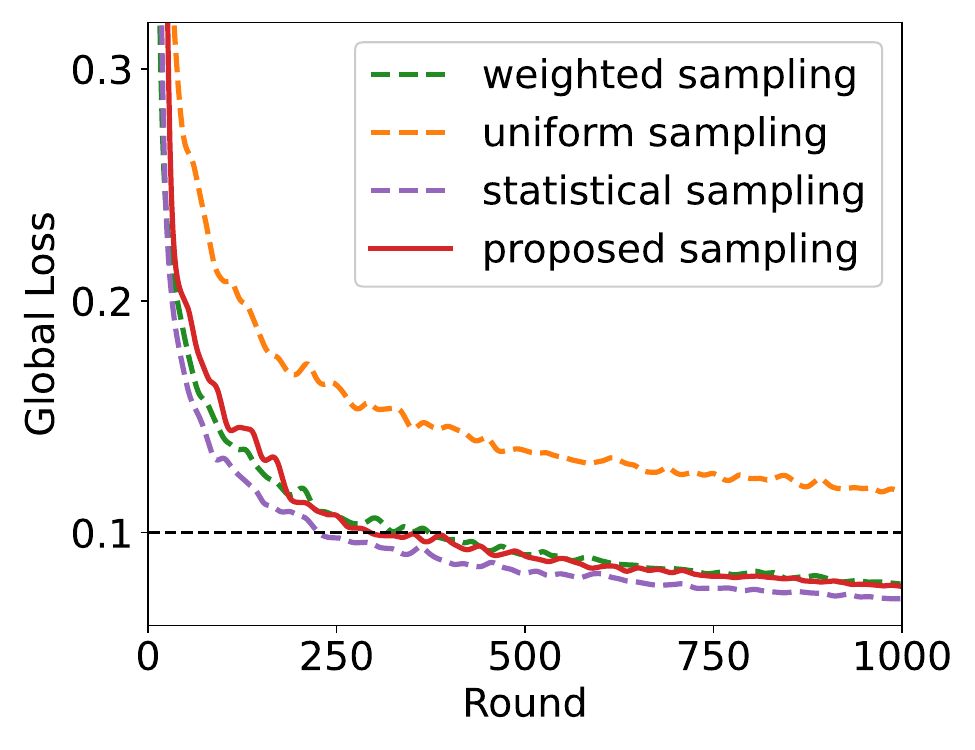}}
\caption{%
{Performances of \textbf{Setup 3} with non-convex CNN and  MNIST dataset for reaching target   %
loss $0.1$ and target accuracy $96.5$\%. %
}}
\label{soft2}
\end{figure*}

\begin{figure*}[!t]
\centering
\label{impact_K}
\subfigure[Loss with wall-clock time]{\label{impact_K_loss}\includegraphics[width=5.97cm]{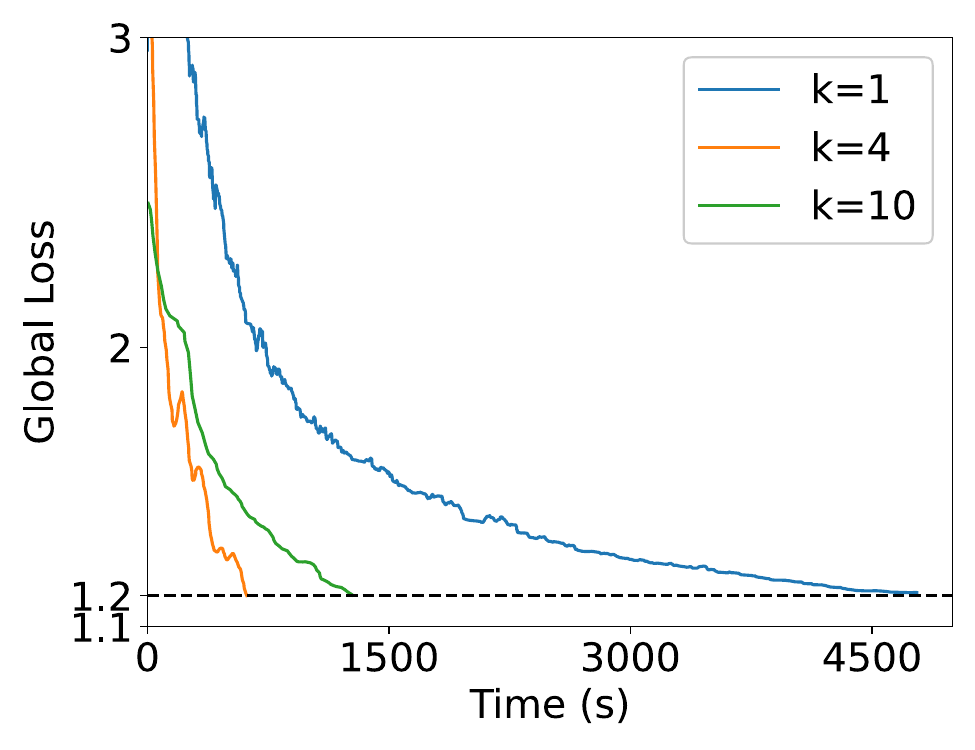}}
\subfigure[Accuracy with wall-clock time]{\label{impact_K_acc}\includegraphics[width=5.97cm]{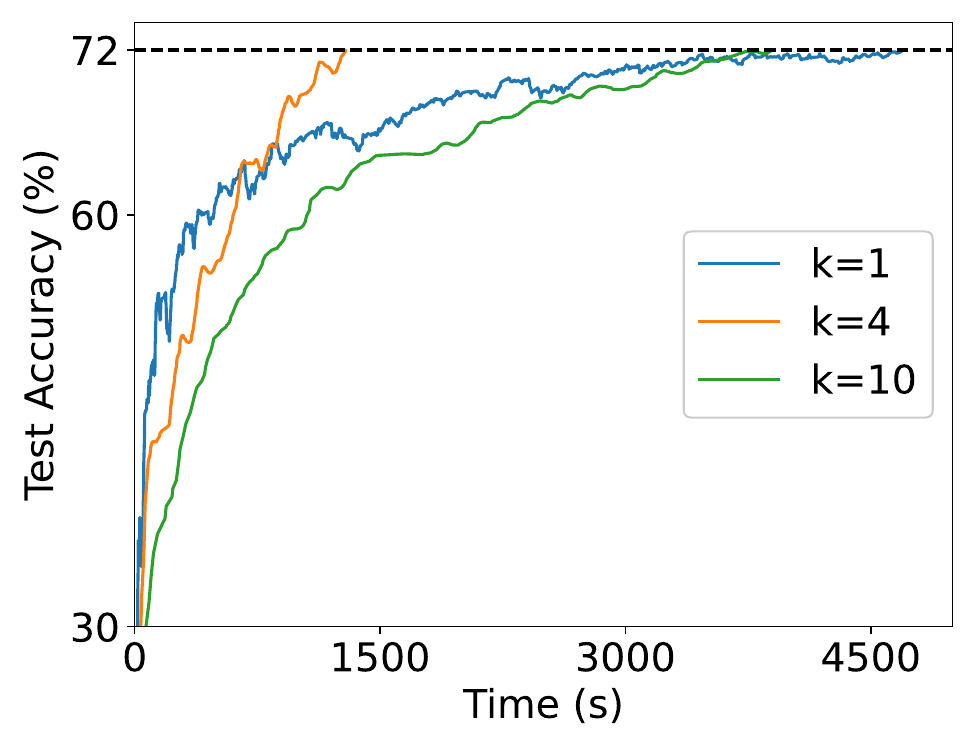}}
\subfigure[Loss with number of rounds]{\label{impact_K_loss_round_syn}\includegraphics[width=5.97cm]{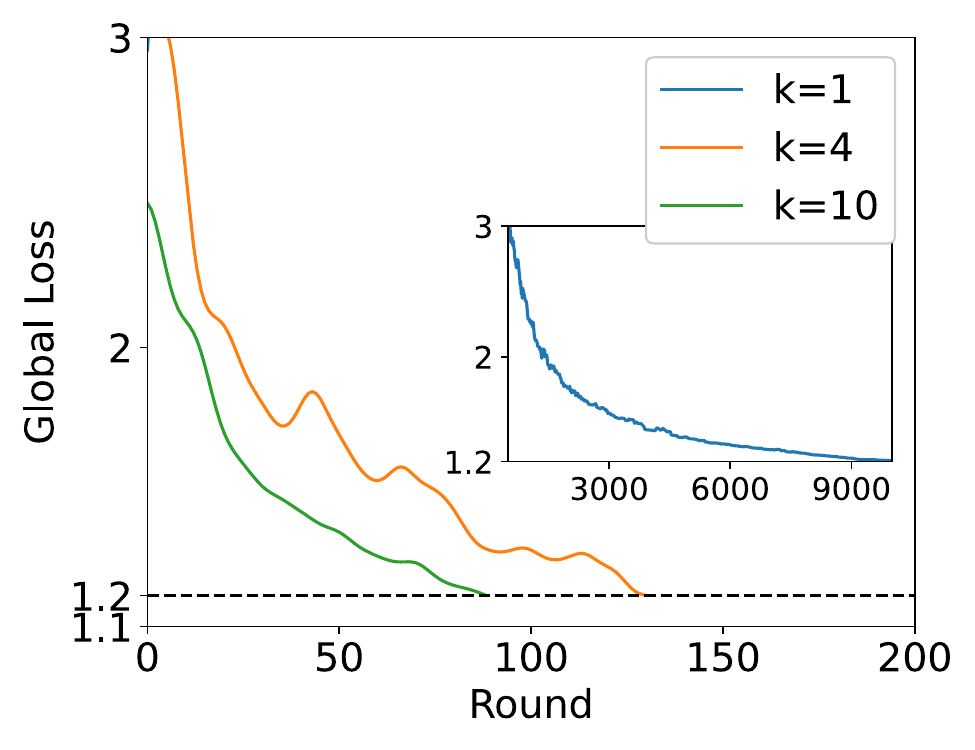}}
\caption{%
Performances of \textbf{Setup 2} of proposed scheme for reaching target loss $1.2$ and target accuracy $72$\% with different sampling number $K$. %
}
\end{figure*}

\item {\emph{Accuracy with Wall-clock Time:} 
As shown in Fig.~\ref{hd}(b), Fig.~\ref{soft}(b), and Fig.~\ref{soft2}(b), our proposed sampling scheme \emph{achieves the target test accuracy %
much faster than 
the other benchmarks}. %
Notably, for Setup $2$ with the target test accuracy $78$\% in Fig.~\ref{soft}(b), our proposed sampling scheme takes around $44$\% less time than uniform sampling, $62$\% less time than statistical sampling, and $60$\% less time than weighted sampling. We can also observe the superior test accuracy performance of our proposed sampling schemes in Prototype Setup $1$ and non-convex Setup $3$ in Fig.~\ref{hd}(b) and Fig.~\ref{soft2}(b), respectively.}

\item {\emph{Loss with Number of Rounds:}  
Fig.~\ref{hd}(c), Fig.~\ref{soft}(c), and Fig.~\ref{soft2}(c) show that  our proposed sampling scheme requires more training rounds for reaching the target loss compared to some of other baselines. This observation is expected since our proposed sampling scheme \emph{aims to minimize the wall-clock time instead of the number of rounds.}} %

\end{itemize}

{\subsubsection{Impact of Client Sampling Number $K$} Fig.~$6$ shows the impact of the client sampling number $K$ for our proposed optimal sampling scheme for Setup~$2$. %
We see from Fig.~6(a) and Fig.~6(b) that as $K$ increases, the time for our proposed optimal sampling scheme reaching the target loss and accuracy \emph{first decreases and then increases}. In particular, the shortest time for reaching the target loss with is $617.1$s under $K=4$, compared to  $4778.1$s, $1272.5$s, and $2021.6$s under $K=1$,$10$, and $16$, respectively. 
  This  demonstrates the trade-off design for the number of $K$ in wireless networks, as a smaller $K$ benefits the per-round time (shown as in Theorem~$2$) but requires more rounds for convergence (shown as in Theorem~$1$), whereas a larger $K$ benefits the number of rounds for convergence but yields a longer per-round time due to  %
   limited wireless bandwidth.} 

\section{Conclusion and Future Work}
\label{sec:conclusion}
In this work, we studied the optimal  client sampling strategy that addresses the system and statistical heterogeneity in FL to minimize the wall-clock convergence time in wireless networks. For system heterogeneity, we proposed an adaptive bandwidth allocation scheme for the sampled clients in each round  with arbitrary sampling probabilities to address the straggling effect. For statistical heterogeneity, we obtained a new tractable convergence bound for FL algorithms with arbitrary client sampling probabilities. Based on the bound, we formulated an analytically non-convex wall-clock time minimization problem. We developed an efficient algorithm to learn the unknown parameters in the convergence bound and designed a low-complexity algorithm to approximately solve the non-convex problem. 
Our solution characterizes the interplay between clients' communication delays (system heterogeneity) and data importance (statistical heterogeneity) on  the  optimal  client sampling design. We also identify a trade-off design for the  number of sampled clients in each round, which balances the impact of wireless bandwidth limitation and convergence rate.  
Experimental results validated the superiority of our proposed scheme compared to several baselines in speeding up wall-clock convergence time. 

{Our client sampling optimization marks an initial effort to tackle system and statistical heterogeneity in FL. Looking ahead, we plan to transform our offline sampling strategy into an online one. This adaptation caters to practical wireless scenarios with changing channel conditions. Furthermore, we're considering a joint optimization of client sampling with other control variables, like the number of sampling instances or local iterations. This joint approach could further cut down total learning time.}

\bibliographystyle{IEEEtran}
\bibliography{ref}

\begin{thebibliography}{10}
\providecommand{\url}[1]{#1}
\csname url@samestyle\endcsname
\providecommand{\newblock}{\relax}
\providecommand{\bibinfo}[2]{#2}
\providecommand{\BIBentrySTDinterwordspacing}{\spaceskip=0pt\relax}
\providecommand{\BIBentryALTinterwordstretchfactor}{4}
\providecommand{\BIBentryALTinterwordspacing}{\spaceskip=\fontdimen2\font plus
\BIBentryALTinterwordstretchfactor\fontdimen3\font minus \fontdimen4\font\relax}
\providecommand{\BIBforeignlanguage}[2]{{%
\expandafter\ifx\csname l@#1\endcsname\relax
\typeout{** WARNING: IEEEtran.bst: No hyphenation pattern has been}%
\typeout{** loaded for the language `#1'. Using the pattern for}%
\typeout{** the default language instead.}%
\else
\language=\csname l@#1\endcsname
\fi
#2}}
\providecommand{\BIBdecl}{\relax}
\BIBdecl

\bibitem{luo2021tackling}
B.~Luo, X.~Li, S.~Wang, J.~Huang, and L.~Tassiulas, ``Tackling system and statistical heterogeneity for federated learning with adaptive client sampling,'' in \emph{Proceedings of the IEEE Conference on Computer Communications (INFOCOM)}, 2022.

\bibitem{chiang2016fog}
M.~Chiang and T.~Zhang, ``Fog and iot: An overview of research opportunities,'' \emph{IEEE Internet of Things Journal}, vol.~3, no.~6, pp. 854--864, 2016.

\bibitem{mao2017survey}
Y.~Mao, C.~You, J.~Zhang, K.~Huang, and K.~B. Letaief, ``A survey on mobile edge computing: The communication perspective,'' \emph{IEEE Communications Surveys \& Tutorials}, vol.~19, no.~4, pp. 2322--2358, 2017.

\bibitem{park2019wireless}
J.~Park, S.~Samarakoon, M.~Bennis, and M.~Debbah, ``Wireless network intelligence at the edge,'' \emph{Proceedings of the IEEE}, vol. 107, no.~11, pp. 2204--2239, 2019.

\bibitem{yang2019federated}
Q.~Yang, Y.~Liu, T.~Chen, and Y.~Tong, ``Federated machine learning: Concept and applications,'' \emph{ACM Transactions on Intelligent Systems and Technology}, vol.~10, no.~2, pp. 1--19, 2019.

\bibitem{kairouz2019advances}
P.~Kairouz, H.~B. McMahan, B.~Avent, A.~Bellet, M.~Bennis, A.~N. Bhagoji, K.~Bonawitz, Z.~Charles, G.~Cormode, R.~Cummings \emph{et~al.}, ``Advances and open problems in federated learning,'' \emph{arXiv preprint arXiv:1912.04977}, 2019.

\bibitem{mcmahan2017communication}
B.~McMahan, E.~Moore, D.~Ramage, S.~Hampson, and B.~A. y~Arcas, ``Communication-efficient learning of deep networks from decentralized data,'' in \emph{Proceedings of the 20th International Conference on Artificial Intelligence and Statistics}, 2017, pp. 1273--1282.

\bibitem{bonawitz2019towards}
K.~Bonawitz, H.~Eichner, W.~Grieskamp, D.~Huba, A.~Ingerman, V.~Ivanov, C.~Kiddon, J.~Kone{\v{c}}n{\`y}, S.~Mazzocchi, H.~B. McMahan \emph{et~al.}, ``Towards federated learning at scale: System design,'' in \emph{Proceedings of Machine Learning and Systems (MLSys)}, 2019.

\bibitem{li2020federated}
T.~Li, A.~K. Sahu, A.~Talwalkar, and V.~Smith, ``Federated learning: Challenges, methods, and future directions,'' \emph{IEEE Signal Processing Magazine}, vol.~37, no.~3, pp. 50--60, 2020.

\bibitem{li2018federated}
T.~Li, A.~K. Sahu, M.~Zaheer, M.~Sanjabi, A.~Talwalkar, and V.~Smith, ``Federated optimization in heterogeneous networks,'' in \emph{Proceedings of Machine Learning and Systems (MLSys)}, 2020.

\bibitem{yu2018parallel}
H.~Yu, S.~Yang, and S.~Zhu, ``Parallel restarted {SGD} for non-convex optimization with faster convergence and less communication,'' in \emph{Proceedings of the AAAI Conference on Artificial Intelligence}, 2019.

\bibitem{yu2019linear}
H.~Yu, R.~Jin, and S.~Yang, ``On the linear speedup analysis of communication efficient momentum {SGD} for distributed non-convex optimization,'' in \emph{Proceedings of the International Conference on Machine Learning (ICML)}.\hskip 1em plus 0.5em minus 0.4em\relax PMLR, 2019, pp. 7184--7193.

\bibitem{wang2018adaptive}
J.~Wang and G.~Joshi, ``Adaptive communication strategies to achieve the best error-runtime trade-off in local-update {SGD},'' in \emph{Proceedings of Machine Learning and Systems (MLSys)}, 2019.

\bibitem{chen2020convergence}
M.~Chen, H.~V. Poor, W.~Saad, and S.~Cui, ``Convergence time optimization for federated learning over wireless networks,'' \emph{IEEE Transactions on Wireless Communications}, vol.~20, no.~4, pp. 2457--2471, 2020.

\bibitem{shi2020device}
W.~Shi, S.~Zhou, and Z.~Niu, ``Device scheduling with fast convergence for wireless federated learning,'' in \emph{Proceedings of the IEEE International Conference on Communications (ICC)}, 2020, pp. 1--6.

\bibitem{nishio2019client}
T.~Nishio and R.~Yonetani, ``Client selection for federated learning with heterogeneous resources in mobile edge,'' in \emph{Proceedings of the IEEE International Conference on Communications (ICC)}, 2019, pp. 1--7.

\bibitem{yang2019scheduling}
H.~H. Yang, Z.~Liu, T.~Q. Quek, and H.~V. Poor, ``Scheduling policies for federated learning in wireless networks,'' \emph{IEEE Transactions on Communications}, vol.~68, no.~1, pp. 317--333, 2019.

\bibitem{luo2020cost}
B.~Luo, X.~Li, S.~Wang, J.~Huang, and L.~Tassiulas, ``Cost-effective federated learning design,'' in \emph{Proceedings of the IEEE Conference on Computer Communications (INFOCOM)}, 2021.

\bibitem{luocostJSAC}
------, ``Cost-effective federated learning in mobile edge networks,'' \emph{IEEE Journal on Selected Areas in Communications}, vol.~39, no.~12, pp. 3606--3621, 2021.

\bibitem{yang2020energy}
Z.~Yang, M.~Chen, W.~Saad, C.~S. Hong, and M.~Shikh-Bahaei, ``Energy efficient federated learning over wireless communication networks,'' \emph{IEEE Transactions on Wireless Communications}, vol.~20, no.~3, pp. 1935--1949, 2020.

\bibitem{haddadpour2019convergence}
F.~Haddadpour and M.~Mahdavi, ``On the convergence of local descent methods in federated learning,'' \emph{arXiv preprint arXiv:1910.14425}, 2019.

\bibitem{karimireddy2019scaffold}
S.~P. Karimireddy, S.~Kale, M.~Mohri, S.~J. Reddi, S.~U. Stich, and A.~T. Suresh, ``Scaffold: Stochastic controlled averaging for on-device federated learning,'' \emph{arXiv preprint arXiv:1910.06378}, 2019.

\bibitem{li2019convergence}
X.~Li, K.~Huang, W.~Yang, S.~Wang, and Z.~Zhang, ``On the convergence of fedavg on non-iid data,'' in \emph{Proceedings of the International Conference on Learning Representation (ICLR)}, 2019.

\bibitem{yang2021achieving}
H.~Yang, M.~Fang, and J.~Liu, ``Achieving linear speedup with partial worker participation in non-iid federated learning,'' \emph{arXiv preprint arXiv:2101.11203}, 2021.

\bibitem{qu2020federated}
Z.~Qu, K.~Lin, J.~Kalagnanam, Z.~Li, J.~Zhou, and Z.~Zhou, ``Federated learning's blessing: Fedavg has linear speedup,'' \emph{arXiv preprint arXiv:2007.05690}, 2020.

\bibitem{stragglers}
R.~Amirhossein, T.~Isidoros, H.~Hamed, M.~Aryan, and P.~Ramtin, ``Straggler-resilient federated learning: Leveraging the interplay between statistical accuracy and system heterogeneity,'' \emph{arXiv preprint arXiv:2012.14453}, 2020.

\bibitem{zhao2015stochastic}
P.~Zhao and T.~Zhang, ``Stochastic optimization with importance sampling for regularized loss minimization,'' in \emph{Proceedings of the International Conference on Machine Learning (ICML)}, 2015, pp. 1--9.

\bibitem{needell2014stochastic}
D.~Needell, R.~Ward, and N.~Srebro, ``Stochastic gradient descent, weighted sampling, and the randomized kaczmarz algorithm,'' in \emph{Advances in neural information processing systems}, 2014, pp. 1017--1025.

\bibitem{alain2015variance}
G.~Alain, A.~Lamb, C.~Sankar, A.~Courville, and Y.~Bengio, ``Variance reduction in {SGD} by distributed importance sampling,'' \emph{arXiv preprint arXiv:1511.06481}, 2015.

\bibitem{stich2017safe}
S.~U. Stich, A.~Raj, and M.~Jaggi, ``Safe adaptive importance sampling,'' \emph{arXiv preprint arXiv:1711.02637}, 2017.

\bibitem{gopal2016adaptive}
S.~Gopal, ``Adaptive sampling for {SGD} by exploiting side information,'' in \emph{Proceedings of the International Conference on Machine Learning (ICML)}.\hskip 1em plus 0.5em minus 0.4em\relax PMLR, 2016, pp. 364--372.

\bibitem{chen2020optimal}
W.~Chen, S.~Horvath, and P.~Richtarik, ``Optimal client sampling for federated learning,'' \emph{arXiv preprint arXiv:2010.13723}, 2020.

\bibitem{rizk2020federated}
E.~Rizk, S.~Vlaski, and A.~H. Sayed, ``Federated learning under importance sampling,'' \emph{arXiv preprint arXiv:2012.07383}, 2020.

\bibitem{nguyen2020fast}
H.~T. Nguyen, V.~Sehwag, S.~Hosseinalipour, C.~G. Brinton, M.~Chiang, and H.~V. Poor, ``Fast-convergent federated learning,'' \emph{IEEE Journal on Selected Areas in Communications}, vol.~39, no.~1, pp. 201--218, 2021.

\bibitem{cho2020client}
Y.~J. Cho, J.~Wang, and G.~Joshi, ``Client selection in federated learning: Convergence analysis and power-of-choice selection strategies,'' \emph{arXiv preprint arXiv:2010.01243}, 2020.

\bibitem{pmlr-v139-fraboni21a}
Y.~Fraboni, R.~Vidal, L.~Kameni, and M.~Lorenzi, ``Clustered sampling: Low-variance and improved representativity for clients selection in federated learning,'' in \emph{Proceedings of the 38th International Conference on Machine Learning}, 2021, pp. 3407--3416.

\bibitem{fraboni2021impact}
------, ``On the impact of client sampling on federated learning convergence,'' \emph{arXiv preprint arXiv:2107.12211}, 2021.

\bibitem{tran2019federated}
N.~H. Tran, W.~Bao, A.~Zomaya, N.~M. NH, and C.~S. Hong, ``Federated learning over wireless networks: Optimization model design and analysis,'' in \emph{Proceedings of the IEEE Conference on Computer Communications (INFOCOM)}, 2019, pp. 1387--1395.

\bibitem{wan2021convergence}
S.~Wan, J.~Lu, P.~Fan, Y.~Shao, C.~Peng, and K.~B. Letaief, ``Convergence analysis and system design for federated learning over wireless networks,'' \emph{IEEE Journal on Selected Areas in Communications}, vol.~39, no.~12, pp. 3622--3639, 2021.

\bibitem{chai2020tifl}
Z.~Chai, A.~Ali, S.~Zawad, S.~Truex, A.~Anwar, N.~Baracaldo, Y.~Zhou, H.~Ludwig, F.~Yan, and Y.~Cheng, ``Tifl: A tier-based federated learning system,'' in \emph{Proceedings of the International Symposium on High-Performance Parallel and Distributed Computing}, 2020, pp. 125--136.

\bibitem{jin2020resource}
Y.~Jin, L.~Jiao, Z.~Qian, S.~Zhang, S.~Lu, and X.~Wang, ``Resource-efficient and convergence-preserving online participant selection in federated learning,'' in \emph{Proceedings of the IEEE International Conference on Distributed Computing Systems (ICDCS)}, 2020.

\bibitem{wang2020optimizing}
H.~Wang, Z.~Kaplan, D.~Niu, and B.~Li, ``Optimizing federated learning on non-iid data with reinforcement learning,'' in \emph{Proceedings of the IEEE Conference on Computer Communications (INFOCOM)}, 2020, pp. 1698--1707.

\bibitem{van2021joint}
N.~Van~Huynh, D.~T. Hoang, D.~N. Nguyen, and E.~Dutkiewicz, ``Joint coding and scheduling optimization for distributed learning over wireless edge networks,'' \emph{IEEE Journal on Selected Areas in Communications}, 2021.

\bibitem{wang2019adaptive}
S.~Wang, T.~Tuor, T.~Salonidis, K.~K. Leung, C.~Makaya, T.~He, and K.~Chan, ``Adaptive federated learning in resource constrained edge computing systems,'' \emph{IEEE Journal on Selected Areas in Communications}, vol.~37, no.~6, pp. 1205--1221, 2019.

\bibitem{tu2020network}
Y.~Tu, Y.~Ruan, S.~Wagle, C.~G. Brinton, and C.~Joe-Wong, ``Network-aware optimization of distributed learning for fog computing,'' in \emph{Proceedings of the IEEE Conference on Computer Communications (INFOCOM)}, 2020.

\bibitem{wang2021device}
S.~Wang, M.~Lee, S.~Hosseinalipour, R.~Morabito, M.~Chiang, and C.~G. Brinton, ``Device sampling for heterogeneous federated learning: Theory, algorithms, and implementation,'' in \emph{Proceedings of the IEEE Conference on Computer Communications (INFOCOM)}, 2021.

\bibitem{bonawitz2016practical}
K.~Bonawitz, V.~Ivanov, B.~Kreuter, A.~Marcedone, H.~B. McMahan, S.~Patel, D.~Ramage, A.~Segal, and K.~Seth, ``Practical secure aggregation for federated learning on user-held data,'' in \emph{NeurIPS Workshop on Private Multi-Party Machine Learning}, 2016.

\bibitem{avent2017blender}
B.~Avent, A.~Korolova, D.~Zeber, T.~Hovden, and B.~Livshits, ``{BLENDER}: Enabling local search with a hybrid differential privacy model,'' in \emph{{USENIX} Security Symposium ({USENIX} Security)}, 2017, pp. 747--764.

\bibitem{konevcny2016federated}
J.~Kone{\v{c}}n{\`y}, H.~B. McMahan, F.~X. Yu, P.~Richt{\'a}rik, A.~T. Suresh, and D.~Bacon, ``Federated learning: Strategies for improving communication efficiency,'' in \emph{NeurIPS Workshop on Private Multi-Party Machine Learning}, 2016.

\bibitem{stich2018local}
S.~U. Stich, ``Local {SGD} converges fast and communicates little,'' in \emph{Proceedings of the International Conference on Learning Representation (ICLR)}, 2018.

\bibitem{lecun1998gradient}
Y.~LeCun, L.~Bottou, Y.~Bengio, and P.~Haffner, ``Gradient-based learning applied to document recognition,'' \emph{Proceedings of the IEEE}, vol.~86, no.~11, pp. 2278--2324, 1998.

\end{thebibliography}

\end{document}